\documentclass[a4paper, 12pt]{article}
\usepackage{fullpage}
\usepackage{sidecap}
\usepackage{changepage}
\usepackage{afterpage}
\usepackage{natbib}
\usepackage{doi}
\usepackage{amsmath, amsthm, amssymb,amsthm}
\usepackage{setspace}
\usepackage{fullpage}
\usepackage{caption}
\usepackage{pdfpages}
\usepackage{algorithm, algorithmic}
\usepackage{amsmath}
\usepackage{parskip}
\usepackage{tikz}
\usetikzlibrary{arrows}
\usepackage{array}
\usepackage{wrapfig}
\usepackage{multirow}
\usepackage{tabularx}
\usepackage{subcaption}
\usepackage{pdflscape}
\usepackage{booktabs}
\usepackage{amsfonts}
\usepackage{dcolumn}
\usepackage{threeparttable}
\usepackage{colortbl}
\usepackage{natbib}
\usepackage{eurosym}
\usepackage{flexisym}
\usepackage{graphicx}
\usepackage[margin=1.1in]{geometry}
\usepackage{ragged2e}
\usepackage{cases}
\usepackage{setspace}
\usepackage{soul}
\usepackage{rotating}
\usepackage{caption}
\usepackage{parskip}
\usepackage[toc,page]{appendix}
\usepackage{rotating}
\usepackage{adjustbox}

\graphicspath{ {./images/} }
\newcolumntype{C}{>{\centering\arraybackslash}p{3.6em}}
\newcolumntype{F}{>{\centering\arraybackslash}p{11em}}

\begin{document}
\onehalfspacing

\newtheorem{thm}{Theorem}
\newtheorem{prop}{Proposition}
\newtheorem{cor}{Corollary}
\newtheorem{lem}{Lemma}
\theoremstyle{definition}
\newtheorem{defn}{Definition}
\newtheorem*{rem}{Remark}

\makeatletter
\def\and{%
  \end{tabular}%
  \hskip 0.3203125em \@plus.17fil\relax
  \begin{tabular}[t]{c}}
\makeatother

\title{\vspace{-1.7453125cm}Cooperation and Cognition in Social Networks\footnote{We thank Peter Bossaerts, Matthew Elliott, and seminar participants at NYU Shanghai and the University of Cambridge for their useful comments. We also thank Wang Song for excellent research assistance. Financial supports from the Cambridge Endowment for Research in Finance (CERF) grant awarded to Edoardo Gallo and the Academic Research Fund (AcRF) Tier 1 Grant RG163/18 from the Ministry of Education of Singapore awarded to Yohanes E. Riyanto are gratefully acknowledged.}}

\newcommand*\samethanks[1][\value{footnote}]{\footnotemark[#1]}

\author{
Edoardo Gallo\thanks{Faculty of Economics, University of Cambridge, Cambridge CB3 9DD, UK; and Magdalene College, CB3 0AG, UK. E-mail: edo@econ.cam.ac.uk.} 
\and Joseph Lee\samethanks
\and Yohanes E. Riyanto\thanks{Division of Economics, School of Social Sciences, Nanyang Technological University, 48 Nanyang Avenue, Singapore 639818. Y.E. Riyanto: e-mail: yeriyanto@ntu.edu.sg; Erwin Wong: e-mail: erwinchi001@e.ntu.edu.sg}
\and Erwin Wong\samethanks
}

\date{May 2023 \\
\vspace{0.05cm}
\vspace{1mm}}
\maketitle

\begin{abstract}

Social networks can sustain cooperation by amplifying the consequences of a single defection through a cascade of relationship losses.  Building on \cite{jackson:2012}, we introduce a novel robustness notion to characterize low cognitive complexity (LCC) networks - a subset of equilibrium networks that imposes a minimal cognitive burden to calculate and comprehend the consequences of defection. We test our theory in a laboratory experiment and find that cooperation is higher in equilibrium than in non-equilibrium networks. Within equilibrium networks, LCC networks exhibit higher levels of cooperation than non-LCC networks. Learning is essential for the emergence of equilibrium play.
\vspace{0.1 in}
\\
\noindent
\textbf{JEL:} C91, C92, D85, Z13.\hfill 

\end{abstract}
\clearpage
\section{Introduction} \label{intro}

A hallmark of a well-functioning society is the ability of its members to cooperate with each other in times of need. Social connections are a natural source of help, and there is extensive empirical support for an association between network structure and cooperative activity (e.g. \cite{greif1993contract}, \cite{putnam2000bowling}). Ample experimental evidence highlights the causal link between the ability to alter the network and cooperation. In particular, social networks matter for cooperation because they enable people to punish defectors by removing connections (e.g. \cite{rand2011dynamic}), and therefore cooperation may thrive in social network structures where there is a threat of punishment through potential link removal. Despite these findings, the majority of theoretical contributions take the network as exogenous and unchangeable (\cite{nava2016networks}). We therefore lack a tight link between theoretical predictions and empirical evidence of how network structure, in combination with the ability to alter the network, affects cooperative activity.

\cite{jackson:2012} (JRT hereafter) develop a model to analyze the emergence of cooperation in networks through the threat of punishment by link removal. In their set-up, agents are embedded in a network and at each point in time one individual is called upon to provide a favor to one of their connections. Her choice is to either provide the favor at a cost to herself, or not provide the favor and lose the connection. JRT characterize the set of networks that are renegotiation-proof equilibria (RPE) -- they are stable because the consequences of removing a connection in terms of the chain of punishment are too big for any agent not to provide a favor. A drawback of this characterization is that the set of RPE networks is rather large. In order to address this, JRT define a notion of robustness against social contagion, and identify a subset of RPE networks that they dub \emph{social quilts}. In these networks, a failure to cooperate leads to a chain of punishment that stays local, and therefore the overall network is robust to such a mistake.

In this paper, we define a novel robustness notion in the context of the JRT framework that is based on the cognitive complexity of figuring out the chain of punishment that is a consequence of removing a connection. Understanding this chain of defections requires iterative reasoning, and there is ample evidence that humans are limited in their ability to do it (\cite{crawford2013structural}). Moreover, empirical and experimental evidence show that individuals struggle to cope with the complexity of network structure and its implications for behavior (e.g. \cite{krackhardt1990assessing}, \cite{dessi2016network}). In our theoretical contribution, we introduce the concept of \emph{cognitive complexity number} to classify RPE networks according to the cognitive burden required to calculate RPE predictions. Networks with minimal cognitive complexity number equal to one are what we define as \emph{Low Cognitive Complexity} (LCC) networks. Such networks collapse to an empty network after a chain of defections resulting from a single link removal. The core idea is that these networks are more robust because it is less likely that any agent removes a link due to a failure in iterative reasoning.

We conduct a lab experiment to test how different network structures (RPE, social quilts, LCC networks) sustain cooperation. Participants are randomly assigned to nodes of a network, and in each period of the game one node/participant is picked at random to either provide a favor or delete one of their links. After the first participant has made a decision, the computer selects another node and the process is repeated until no links are left or all participants have consecutively chosen not to delete any link. The experiment has two phases. In Phase 1, participants play the above game on two different networks in succession with the groups reshuffled every time they get assigned to a new network. Between Phase 1 and 2, participants face decision-making tasks against computer players. In Phase 2, participants play the game three times on the two different networks used in Phase 1 with groups again reshuffled every time they are assigned to a new network. Additionally, we elicit social preferences using a Dictator Game.

The first result is that participants are less likely to delete links in RPE networks compared to links in non-RPE networks. This shows that the RPE notion in JRT has behavioral validity in predicting which network structures are more conducive to cooperative activity. Restricting the attention to social quilts, however, provides no additional benefit in terms of cooperation. In other words, participants in a social quilt are as likely to delete links as participants in an RPE network that is not a social quilt. This shows that robustness against social contagion is not a primary determinant of cooperative activity in the networks examined in the experiment.

The second result is that participants are less likely to delete links in LCC networks compared to social quilts and/or other RPE networks that are not LCC. This validates our hypothesis that bounded rationality considerations are a primary driver of what types of network structures are conducive to cooperation.

The third result is that learning matters. Participants need time to understand the consequences of defection and are better able to do so in networks with low cognitive complexity, further showing how cognitive complexity is an important determinant of sustaining cooperation. In Phase 1, social preferences are the primary determinant of link removal decisions, while the type of network structure is not significant. In Phase 2, however, after participants have had the chance to learn from playing against computer players, network structure becomes a significant determinant of link removal decisions and social preferences are insignificant. Moreover, participants are less likely to remove links in LCC compared to any other type of network, including social quilts.

The rest of this paper is organized as follows. Section \ref{lit} discusses the related literature. Section \ref{theory} presents the theoretical results, and the resulting hypotheses. Section \ref{design} explains the set-up of the experiment. Section \ref{results} presents the results of the study, and section \ref{conclusion} concludes. The main Appendix contains the proofs, a more general version of the theoretical model, and the instructions of the experiment.

\section{Related literature} \label{lit}

Our work is related to several strains of theoretical and experimental research in repeated games, networks, and behavioral economics. After a brief overview of related work outside of economics, we review each one of these in turn and discuss how this paper draws novel connections across these areas. 

The question of how social structure determines cooperative behavior is a long-standing interest in many academic disciplines. In sociology, seminal work by \cite{simmel1950sociology} and \cite{coleman1988social} argues that transitivity or closure -- whether someone's friends are friends with each other -- is crucial for cooperation. Their contributions and subsequent work, however, lack firm theoretical foundations (\cite{sobel2002can}). A vast empirical literature shows an association between cooperation and social network characteristics in many disciplines including sociology (e.g. \cite{putnam2000bowling}), management (e.g. \cite{burt2007brokerage}), anthropology (e.g. \cite{apicella2012social}), and political science (e.g. \cite{bond201261}). The endogeneity of the network coupled with the difficulty of separating the impact of network structure from unobservable individual characteristics (\cite{manski2000economic}), however, means that it is challenging for these contributions to identify how network structure causally affects cooperative behavior.

Due to the difficulty of combining network formation with a prisoner's dilemma type of game to capture cooperation, the majority of theoretical contributions in the economics literature take the network as exogenous. For instance, \cite{haag2006social} show that the optimal network to sustain cooperation depends on the distribution of preferences, and completely connected cliques are optimal when there is enough heterogeneity.\footnote{\cite{wolitzky2013cooperation} defines a novel centrality measure to show that more central players cooperate more in the equilibrium with the maximum level of overall cooperation. Other contributions include  \cite{bloch2008informal} and \cite{lippert2011networks}. See \cite{nava2016networks} for a comprehensive review of this literature.} \cite{jackson:2012} is the first paper to incorporate a reduced model of network formation with a cooperation game on a network. Section \ref{theory} describes their framework in detail because the theoretical model in our paper is a reduced form version of theirs. The main result in \cite{jackson:2012} is to characterize the set of RPE networks, and identify a subset of social quilts networks that are robust against social contagion. We build on their theoretical work by introducing a novel notion to select among RPE networks based on bounded rationality considerations, and characterize this equilibrium subset. The main contribution of our paper is to test experimentally the behavioral validity of these equilibria to sustain cooperation.

Early experimental work on cooperation in networks also focused on the case where the network is exogenous. \cite{cassar2007coordination} finds no significant difference in cooperation levels across three types of networks of 8 nodes. Subsequent work confirmed this finding on networks with different structure and sizes.\footnote{In particular, \cite{kirchkamp2007naive} find no substantial differences across circle networks of different sizes, while \cite{gracia2012heterogeneous} find no difference between two large networks. An exception is \cite{rand2014static} who find some effect of network structure if the benefits from cooperation are large enough.} A more recent experimental literature endogenizes the network and shows that the ability to punish defectors by removing connections and/or rejecting proposals to connect fosters cooperation (\cite{rand2011dynamic}, \cite{wang2012cooperation}).\footnote{The efficacy of this network punishment mechanism depends on information about others' reputation (\cite{gallo2015effects}, \cite{cuesta2015reputation}), knowledge of the social network structure (\cite{gallo2015effects}), accuracy of the monitoring technology (\cite{gallo2022cooperation}), and the strength of the connections (\cite{shirado2013quality}, \cite{gallo2019strong}). See \cite{choi2016networks} for a comprehensive review.} All these prior works lack a theoretical framework, while our experiment is explicitly designed to test precise theoretical predictions. 

Experimental (\cite{dessi2016network}) and empirical (\cite{breza2018seeing}) evidence show that humans have a biased and limited perception of social network structure. Insights from behavioral economics seem, therefore, to be quite relevant to understanding how social network structure affects behavior. Despite this clear link, however, there are very few papers at the intersection of behavioral economics and networks.\footnote{Exceptions include \cite{ushchev2020social} who incorporate conformity in network games, and \cite{gallo2020social} who study the effect of confirmation bias on social learning.} 

In particular, we examine how individuals' bounded rationality helps in identifying a subset of equilibria that predict how networks affect behavior. Several empirical and experimental studies show that humans engage in iterative thinking to solve strategic problems (\cite{crawford2013structural}). Structural models based on cognitive hierarchy (\cite{camerer2004cognitive}) and/or level-k reasoning (e.g. \cite{nagel1995unraveling}, \cite{costa2001cognition}) have been successfully applied to explain behavior in, amongst others, auctions (\cite{crawford2007level}), market entry (\cite{goldfarb2011thinks}), and strategic communication (\cite{crawford2003lying}). In this paper we apply the logic of iterative thinking to identify a subset of RPE networks that require only one step of iteration to figure out the full consequences of a defection, and we show that cooperation is higher in these networks compared to other equilibria.

We are the first paper, to our knowledge, to investigate the role of bounded rationality in sustaining cooperation. Given that equilibrium multiplicity is ubiquitous in network problems\footnote{Examples include network formation (\cite{jackson1996strategic}), public goods games (\cite{bramoulle2014strategic}), and social learning (\cite{bala1998learning}).}, this insight may be applicable to other contexts where networks matter.

\section{Theory} \label{theory}

\subsection{Model} \label{model}
The theoretical framework examined in this paper is a reduced form version of the model in \cite{jackson:2012} (JRT hereafter). This section introduces relevant network concepts and the reduced form game. At the end we discuss the main simplifications to the JRT framework.

\medskip
\noindent
\textbf{Networks.} A finite set $N = \{1, . . . , n\}$ of agents is connected by a social network represented by an \emph{unweighted, undirected graph} $g$. $ij$ denotes a link between $i$ and $j$, and the graph $g - ij$ represents the network obtained from $g$ by deleting link $ij$.\footnote{Throughout this paper, we assume there is no self-link, ie $ii=0$ for all $i$.} The \emph{neighborhood} of agent $i$ in $g$ is $N_i(g)=\{j|ij \in g\}$. $d_i(g)\equiv |N_i(g)|$ denotes the \emph{degree} of $i$ in $g$. $D(g)$ denotes the profile of degrees associated with a network g, i.e. $D(g) \equiv (d_1(g),...,d_n(g))$. We say that $D(g)>D(g')$ if $d_i(g)\geq d_i(g')$ for all $i$ with strict inequality for some $i$. The empty network is a network $g^{\empty}$ in which there are no links. A \emph{complete network} is a network in which $ij \in g$ for all $i,j \in{N}$. A \emph{triangle network} is a complete network of three nodes.\footnote{For the purpose of this paper, a complete network of three nodes with other isolated nodes is also called a ``triangle network''.} We say $g'$ is a \emph{subnetwork} of $g$ if $ij \in g'$ only if $ij \in g$. A \emph{walk} in a network $g$ refers to a sequence of nodes $i_1$, $i_2$, ... $i_K$ such that $i_k i_{k+1} \in g$ for each $k$ from $1$ to $K-1$. A \emph{simple cycle} in a network $g$ is a walk in $g$, $i_1$, $i_2$, ... $i_K$, such that $i_1 = i_K$, and the only repeated vertices are $i_1$ and $i_K$.

\noindent
\textbf{Game.} Let $t = \{0, 1, ...\}$ denote discrete time periods. The game begins with some initial network $g_0$ at time $t=0$. Consider time period $t$ and $ij\in g_t$, there is a probability $p$ that $i$ will request a favour from $j$ at $t$. Let $n(n-1)p \leq 1$ and the time between periods be small so that we can assume that at most one favour is needed across all agents in any given period. As JRT discuss in more detail, this is a proxy for a Poisson arrival process.

When an agent $j$ is called to do a favour for agent $i$, agent $j$ can then choose to either perform the favour or refuse the favour. Performing the favour costs agent $j$ an amount of $c>0$, it provides agent $i$ a benefit $v > c$, and involves no change in the network. Refusing the favour entails the permanent loss of the $ij$ link in future periods so $g_{t+1}=g_t-ij$, but there are no costs/benefits for either agent. Assuming agents discount future payoffs according to a $\delta \in (0,1)$ factor, an agent $i$'s expected utility from having a link $ij \in g$ in perpetuity is therefore $b\equiv p(v-c)/(1-\delta)$. This represents the present discounted value of a perpetual mutual favour-exchange relationship. We assume throughout this paper that $2b > c > b$. Appendix A discusses the generalization of the analysis to all $c$ and $b$. Note that this is a complete information game, i.e. agents are fully aware of all moves in the game at every node.

\noindent
\textbf{Reduced Form.} We have adopted a reduced form of the JRT model along the following two dimensions in order to simplify exposition and match our experimental set-up. 

First, in the original JRT framework agents simultaneously announce the links that they are willing to retain in each time period $t$ prior to favour decisions. There is, however, no strategic benefit for agent $i$ to remove her link with another agent $j$ because retaining the link is costless and allows agent $i$ to potentially call on agent $j$ for a favour in the future. We have, therefore, suppressed this stage both in the theory and the experiment to focus on the favour exchange stage.

Second, we restrict our attention to the $2b > c > b$ range. The lower bound is identical to JRT and allows the focus to be on the interesting case where simple bilateral relations cannot sustain cooperation. The upper bound simplifies the characterization of equilibrium networks, and it is consistent with the calibration in the experiment.

\subsection{Results} \label{theoryresults}
The first part of this section summarizes the characterization of equilibrium predictions in JRT in the context of the reduced form model. Following their approach, it characterizes the subset of \emph{social quilts} equilibria that are robust against social contagion. In the second part we introduce the novel notion of \emph{cognitive complexity} to provide an alternative method to select equilibria robust against social contagion based on behavioural considerations. 

JRT analyze the game using the renegotiation-proof equilibrium (RPE hereafter) concept, and we refer the reader to their paper for a detailed discussion on their definition of RPE in this game. Here we briefly restate their definitions and focus on results that are relevant to formulate our hypotheses in section \ref{hyp} within our reduced form model. In particular, the assumption that $2b > c > b$ is crucial to simplify the ensuing notation, statements, and exposition. 
First, the following is the adaptation of the definition of Transitively-Critical (TC) networks to our reduced model.

\begin{defn}
\label{defn:tc}
The class of \textbf{TC networks} are generated as follows: Define $k$ as the number of links in the network. For $k=0$, define $TC_0=\{g^{\emptyset}\}$. Inductively on $k$, $TC_k$ is such that $g \in TC_k$ if and only if $\forall i,ij \subset g, \exists g' \subset g-ij $ such that 
\begin{itemize}
 \item $g' \in TC_{k'}$ for some $k'\leq k-2$; 
 \item $d_i (g') \leq d_i (g)-2$
 \item $ \nexists g'' \in TC_{k''} $ s.t. $ g'' \subset g-ij $ and $ D(g'')>D(g') $.
\end{itemize}
\end{defn}

\noindent
Intuitively, the TC concept states that in all continuing sub-networks that are also TC, individuals who lost links must have lost at least 2 links. A key result in JRT is to show that there is a one-to-one correspondence between RPE and the set of TC networks.

\begin{prop}
\label{prop:RPE}
\textbf{(JRT)} A network is a RPE if and only if it is Transitively Critical.
\end{prop}

The intuition is that in deciding whether to accept or refuse the favour, agents consider the eventual resulting network from refusing the favour. If the loss of favour-exchange relationships is more than the cost of the favour (i.e. at least 2 links to the individual are deleted as a consequence of not performing the favour given $2b>c>b$), then the agent will do the favour. If it is optimal for all agents to do the favours for each of their links, then the network is a RPE network.\footnote{See Appendix in JRT for a detailed proof.}

A drawback of Proposition \ref{prop:RPE} is its limited predictive power because the class of TC networks is quite large. Additionally, some TC networks are not robust to the random deletion of a link due to an accident and/or failure by an agent to calculate RPE predictions. JRT introduce a stronger equilibrium concept of ``robustness against social contagion'' within the class of TC networks, where the deletion of a link results in ``local'' rather than ``global'' damage. These networks prevent the entire society from unravelling if an agent deletes links contrary to RPE predictions.

\begin{defn}
\label{defn:robust}
A network $g$ is \textbf{robust against social contagion} if it is RPE and sustained as part of a pure strategy subgame perfect equilibrium with $g_0 = g$ that satisfies the following condition. In any subgame continuation from any RPE $g' \subset g$, and for any $i$ and $ij \in g'$, if $i$ deletes the link with $ij$, then the continuation leads to $g''$ such that if $hl \not\in g''$ then $h \in N_i(g') \cup \{ i \}$ and $l \in N_i(g') \cup \{ i \}$.
\end{defn}

For networks that are robust against social contagion, if an agent deletes links contrary to RPE predictions then the only link deletions in steady-state involve that agent and her neighbours. One of the key results in JRT is to show that there is a one-to-one correspondence between networks that are robust against social contagion and a class of networks that they dub ``social quilts''.

\begin{defn}
\label{defn:quilt} 
A \textbf{social quilt} is a union of triangle networks such that there is no simple cycle in the network involving more than 3 nodes.
\end{defn}

\begin{prop}
\label{prop:quilt}
\textbf{(JRT)} A network is robust against social contagion if and only if it is a social quilt. 
\end{prop}

Intuitively, should an agent delete a link, only the triangle network that the link is part of will be lost. The other triangle networks remain untouched as each triangle network is RPE and is self-sustaining. Therefore, only the agent deleting the link and her neighbours will be affected, and the damage remains localised.

Social quilts provide a mechanic way to stem contagion by limiting it to a small part of the network. This is an appealing solution if a random link deletion is an event purely caused by chance, and any link is equally susceptible to this. In reality, however, mistakes that lead to link deletions may be the result of an agent's bounded rationality that makes her unable to calculate RPE predictions. If bounded rationality is the cause of deletion mistakes then deletions will depend on an agent's level of (bounded) rationality as well as how difficult it is to compute RPE predictions in a given network environment.

Here we propose an alternative notion to select among RPE equilibria based on considerations of bounded rationality. In particular, we define the concept of ``cognitive complexity number'' to classify RPE networks according to the cognitive burden required to calculate RPE predictions, and focus on the class of networks that have cognitive complexity number equal to 1.

\begin{defn}
\label{defn:lowCC}
Denote by $cc(g)$ the \textbf{cognitive complexity number} of network $g$, and assume $cc(g^{\emptyset})=0$. Suppose we have a TC network $g$ and the largest TC subnetwork of $g$ is $g' \subset g$, then $cc(g)=cc(g^{'})+1$. We define $g$ to be a \textbf{Low Cognitive Complexity (LCC)} network if $cc(g)=1$.
\end{defn}

Given that computing RPE predictions requires the calculation of TC networks, the notion of cognitive complexity number is based upon the definition of TC networks. In particular, an agent needs to ascertain the largest TC subnetwork should she refuse a favour. If the subnetwork is not an empty network, the agent must iteratively identify further subnetworks until she reaches the empty network. The more iterations she has to go through, the more cognitively complex the network is. We define LCC networks to be the ones that collapse to the empty network after one link removal because from a cognitive point of view these are the simplest environments to understand the consequences of a favour refusal. The following proposition characterizes the class of LCC networks.

\begin{prop}
\label{prop:lowCC}
A network is LCC if and only if it is a simple cycle.
\end{prop}

The intuition is as follows. If an agent only has one link connected to her, she will refuse any favour when called upon, since the cost of doing a favour is more than any potential benefit from retaining the link. This means that if any link is broken in a simple cycle, it will trigger a series of deletions and the resultant network will be the empty network. Thus, a simple cycle is LCC. Since agents with links have two or more links, all non-empty RPE networks minimally contain a simple cycle. As such, only simple cycles are LCC, since all other networks will have at least a simple cycle as the resultant network instead of the empty network. An implication of this result is that the simple cycle can be thought of as a ``basic unit'' to help us generate RPE networks with higher cognitive complexity number. Appendix A contains a detailed proof.

Within the $2b>c>b$ parameter range, there is one RPE network that satisfies both equilibrium concepts of LCC and robustness against social contagion.

\begin{cor}
\label{cor:triangle}
The triangle network is the only network which is LCC and SQ.
\end{cor}

This follows from Definition \ref{defn:quilt} and Proposition \ref{prop:lowCC}, since triangle networks are both social quilts and simple cycles. Other simple cycles involving more than 3 nodes would no longer be social quilts by Definition \ref{defn:quilt}.

\subsection{Hypotheses} \label{hyp}

This section lists and discusses the hypotheses we are investigating in our experiment. The basic check of the predictive power of Proposition \ref{prop:RPE} is that cooperation is higher in RPE networks compared to networks that are not RPE. The top part of Figure \ref{allnetworks} shows the six RPE networks we investigate in our experiment, while the bottom part of Figure \ref{allnetworks} displays the five non-RPE networks that we compare them to. In the set-up of our experiment, cooperation means not deleting links, so the following is our first hypothesis.

\begin{figure}[h]
    \centering
    \caption{Networks in our experimental design}
    \includegraphics[width=\linewidth]{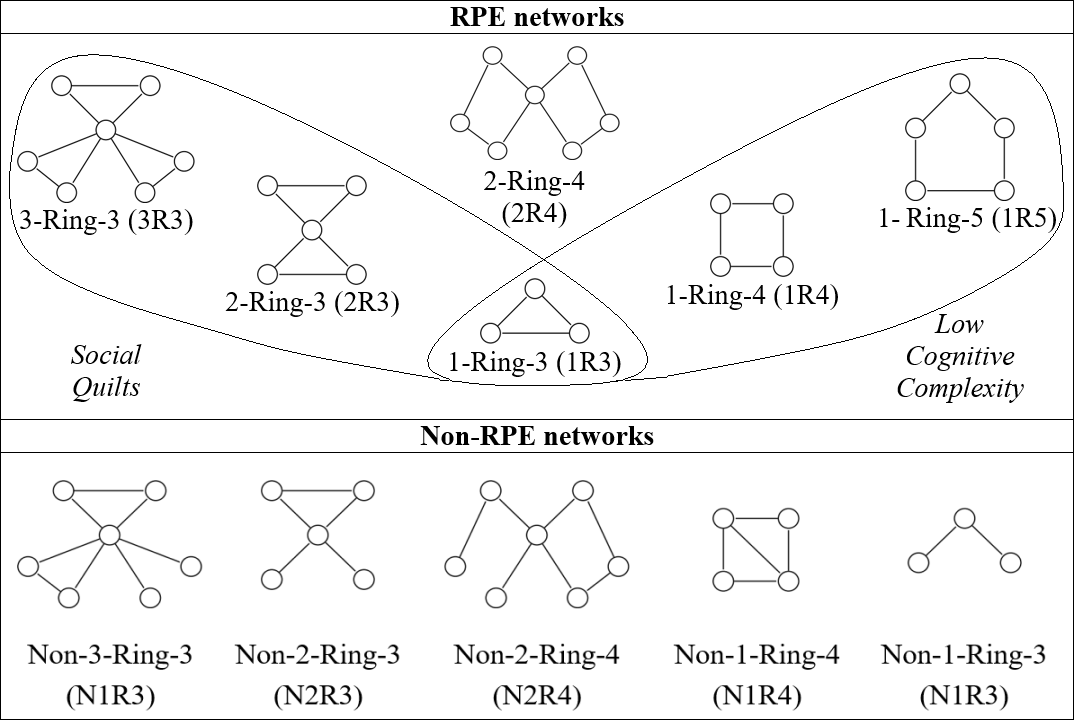}
    \label{allnetworks}
\end{figure}

\vspace{2mm}
\noindent
 \emph{H1: Participants are less likely to delete links in an RPE compared to a non-RPE network.}

As we discuss in section \ref{model}, the class of RPE networks is quite large so the predictive power of this equilibrium concept is limited. For instance, it does not discriminate among the six RPE networks in our experiment despite their evident differences in structure. JRT propose social quilts as a refinement to identify a set of networks where cooperation is more likely to be maintained because they are robust to contagion. Given that contagion stays ``local'' in SQ networks, participants ought to be less likely to delete links because there is a lower probability that they will be involved in a contagion chain compared to participants in other RPE networks that are not SQ. In our design, we investigate three RPE networks that are social quilts. They are composed by one, two or three rings of three nodes with a central node that belongs to all the rings. These are labelled as $x$-Ring-3 with $x\in \{1,2,3\}$ as shown in the top-left part of Figure \ref{allnetworks}. This leads to our second hypothesis.


 \vspace{2mm}
\noindent
 \emph{H2: Participants are less likely to delete links in an SQ network than in an RPE network that is not SQ. In particular, participants are less likely to delete links in x-Ring-3 networks compared to x-Ring-y networks, where $x\in \{1,2,3\}$ and $y>3$.}
 
In section \ref{model} we propose the novel concept of LCC networks to identify a subset of RPE networks where cooperation is more likely to be sustained. The core idea is that participants are more likely to cooperate in networks whose structure makes it cognitively easy to see the chain of contagion that is a consequence of a failure to cooperate. This is a different approach to carve out a set of robust RPE networks compared to the concept of social quilts. In the context of our design, we investigate three RPE networks that are LCC formed by one ring of three, four and five nodes. For this reason, we label them 1-Ring-y with $y\in {3,4,5}$ as shown in the top-right part of Figure \ref{allnetworks}. Note that, as shown in Figure \ref{allnetworks}, 1-Ring-3 is the only RPE network in our experiment that is both a SQ network and a LCC network at the same time. This leads to our third hypothesis.


 \vspace{2mm}
\noindent
 \emph{H3: Participants are less likely to delete links in an LCC network than in an RPE network that is not LCC. In particular, participants are less likely to delete links in 1-Ring-y networks compared to x-Ring-y networks, where $x>1$ and $y \in \{3,4,5\}$.}
 
\section{Experimental design and procedures} \label{design}

The experiments were conducted at Nanyang Technological University (NTU), and programmed using oTree (\cite{chen:2016}). A total of 225 NTU students from all majors participated in our experiments. At the beginning of each session, we ensured anonymity by assigning each participant a random numerical ID as their identifier throughout the experiment. Participants had a physical copy of the instructions, and the instructions were also displayed on the computer screen as they were read aloud by the experimenter. In the rest of this section, we describe the game and explain the overall workflow of the experiment.

In each network game, subjects in the session are split into groups according to the size of the networks they play in.\footnote{If the number of participants is not divisible by the network size ``excess'' players are chosen at random to sit out.} Each subject in the group is then randomly assigned to a node in the network. In each round, the computer randomly selects one of the nodes, and the subject assigned to that node has to decide to delete one of their links or to keep all of their links.\footnote{If the participant does not make a choice then the computer automatically selects no link deletion after 60 seconds.} After the first participant has made a decision, the computer selects another node and the process is repeated for the next round. The sequence of node selection is randomized, with the constraint that all other nodes have to be selected before a node is selected a second time. Participants can only see the network at the time they need to make a decision, and they do not know the sequence of node selection nor the history of other participants' decisions.

As it is not practical for the experiment to model time-discounting flows, the design cannot exactly replicate favour exchanges as per the theoretical set up in section \ref{model}. Instead of having a favour provide a benefit $v$ to the recipient at the cost of $c$ to the giver, we proxy this by participants receiving a one-time benefit of $b$ for each of the links connected to them that survive at the end of the game. To maintain the dynamics and the RPE predictions in section \ref{model}, we allow participants to gain $c$ for each link they delete, which is the relative gain they would get from refusing to perform a favour (since they would save on the cost $c$).

However, this new structure introduces a new strategic element not present in the model in section \ref{model}. As an example, consider a triangle network with participants $\{i,j,k\}$, and assume $i$ refuses a favour to $k$. In the JRT framework, the $ik$ link breaks, the network is no longer RPE, and therefore it immediately collapses to the empty network so $i$ immediately loses two links as a result of her refusal. Consider, instead, our experimental set-up. When $i$ removes the link and gains $c$, there are still opportunities for participants to gain a total of $2c$ collectively from deleting the remaining two links. This presents a possible corner case. Suppose the order of turns is $i\rightarrow j \rightarrow k \rightarrow i$. Suppose that $i$ decides to cut the link with $k$. After $i$'s decision, it is $j$'s turn to decide which link to be deleted, and $j$ may choose between deleting against $i$ or $k$. If $j$ deletes the $jk$ link, then $k$ has no link to delete in her turn, and $i$ can delete the $ij$ link in the next turn. While the probability $h$ of this scenario occurring is low, we need to set $2b > c + hc > b$ to account for this corner case. By setting $c$ to be close to $b$, we ensure that participants who are risk-averse, risk-neutral or even slightly risk-loving will still find it optimal to play according to equilibrium predictions.

In order to mitigate the impact of this strategic effect, we chose $b=100$ points for each of the participants' links that survives deletion, and $c=110$ points for each link they delete. These parameters are within the $2b > c > b$ range in our main model and they account for the above corner case, given that $c$ is close to $b$.

Another difference between this game and the theoretical set-up in section \ref{model} is that nodes, rather than links, are selected for decisions. The reason is that in the (relatively) small networks in the experiment selecting a link would imply significant differences across participants in the likelihood of being selected to make a decision (e.g. the center node in 3-Ring-3 would be three times as likely to be picked than every other node). It would, therefore, be challenging to set-up a randomized link selection process that would not give a participant more opportunities to play the game in a way that is excluded in the theory.

The workflow of the experiment is described in Table \ref{table:flow}. In Phase 1, participants play the above game on two different networks. After the first network, the groups are reshuffled to minimise spillovers across games and subsequently they play the game on the second network. Following Phase 1, participants face an individual decision-making task against computer players described in detail below. In Phase 2, participants play the game three times on two different networks used in Phase 1 with the first and last plays on the second network seen in Phase 1, and the second play on the first network seen in Phase 1. To account for order effects, we run another experiment session with the opposing order of networks seen in Phase 1. Similarly to Phase 1, the groups are reshuffled after each network.

\begin{table}
\centering
\begin{threeparttable}
\caption{Tasks, participants and rounds by treatment group.}
\centering
\arrayrulecolor{black}
\label{table:flow}
\begin{tabularx}{\textwidth}{@{} F *{7}{X} @{}}
\toprule
    & & & \multicolumn{2}{c}{Treatment} & & & Totals\\
\cmidrule{2-7}
    & 1 & 2 & 3 & 4 & 5 & 6 \\
\toprule
\toprule
    & \multicolumn{6}{c}{Networks} &\\
    & & & & & & \\
    Phase 1 & 1R3, N1R3 & 1R4, N1R4 & 1R5, 2R3 & 2R4, 3R3 & 2R3, 1R3 & 2R3 \\
    & & & & & & \\
    Computer Task & \multicolumn{6}{c}{All participants play each network in Figure \ref{allnetworks}} &\\
    & & & & & & \\
    Phase 2 & N1R3, 1R3, N1R3 & N1R4, 1R4, N1R4 & 2R3, 1R5, 2R3 & 3R3, 2R4, 3R3 & 1R3, 2R3, 1R3 & 2R3 \\
    & & & & & & \\
\bottomrule
\toprule
    Participants & 48 & 40 & 50 & 42 & 30 & 15 & 225\\
\bottomrule
\toprule
\end{tabularx}
\begin{tablenotes}
\footnotesize
\item Note: There are 2 sessions per treatment. Given the same treatment, the only difference between the sessions is that the orders of networks seen in Phase 1 and Phase 2 are reversed. For each treatment, the specific networks, order of networks played in each round, and number of participants are specified.
\end{tablenotes}
\arrayrulecolor{black}
\end{threeparttable}
\end{table}

Between Phase 1 and 2 participants play the same game against computer players, which we call the Computer Task. They are told that each of the other nodes is a computer player, and that all computer players make optimal decisions. Specifically, we programmed the computer player not to delete any link if the network is RPE, and to delete one of its links if the network is not RPE and this link should be deleted in equilibrium. This game is repeated for the eleven different networks in Figure \ref{allnetworks} so in each session we can compare how a participant plays against other participants as opposed to computer players for each network.\footnote{In the Computer Task, the participant is always assigned to (one of) the node(s) with the highest degree. When there are multiple nodes with the highest degree, then the assignment is randomized.}

There are two reasons to make participants play against computer players. First, it generates useful data to disentangle social preferences from limited cognitive ability as a driver of decisions. When a participant fails to remove a link in a non-RPE network against other participants, it can be due to the inability to understand the equilibrium or because of altruism -- the latter channel is clearly suppressed when playing against a computer. Second, it allows us to disentangle whether a participant removes a link in a non-RPE network because of limited cognitive ability or because it believes other participants have limited cognitive ability (and therefore it is rational to delete the link because it would not survive anyway). The latter driver of decisions is clearly absent with computer players who always make rational (optimal) decisions, and therefore this task allows us to partially identify the cognitive sophistication of participants by exposing them to networks with different cognitive complexity numbers.

At the very beginning of the experiment, we elicit social preferences using a Dictator Game. Each participant has to decide how much of an endowment of 1 SGD to allocate to another randomly selected participant. We ensure that a participant is never a recipient of her recipient's allocation to exclude reciprocity as a factor influencing their allocation decisions. The outcome of the game is only communicated to participants at the end of the experiment to ensure earnings do not affect decisions in the experiment. 

Toward the end of the experiment, we elicit participants’ risk preferences using the Holt-Laury multiple-price list risk elicitation (\cite{holt:2002}). We also conduct a post-experimental questionnaire to collect data on individual characteristics such as game theory experience, age, and gender. On top of their earnings from the Dictator Game and the Holt-Laury risk elicitation, one round from Phase 1 and two rounds from Phase 2 are randomly selected to determine the total payment of each participant. The experimental sessions lasted an average of 75 minutes, and the average earning for each participant was 21 SGD. 

\section{Empirical Results} \label{results}

We begin by focusing our regression analysis at the decision-level. Every time participants are assigned to a network, they have to make a decision on whether to keep all links or delete one of them. For example, consider the case where three participants have been assigned to play in the 1R3 network. Clearly, there will be a minimum of 3 decisions made in this network game, and therefore this network game will contribute to a minimum of three units of observation to the decision-level data. We run logit regressions, where each unit of observation is a decision made by a participant, with standard errors clustered at the subject level in the following form.
\[
Y_{i, s} =  \beta_{0} + \beta_{T}\textbf{T}_{i} + \beta_{C}\textbf{C}_{s} + \epsilon
\]
The dependent variable $Y_{i, s}$ is binary -- it equals 1 if decision $i$ made by subject $s$ is to delete a link, and equals 0 otherwise. $\textbf{T}_{i}$ denotes a vector of indicator variables capturing whether the network at which decision $i$ is made is an RPE network, a SQ network, a LCC network, or a 1R3 network. $\textbf{C}_{s}$ denotes a vector of control variables unique to subject $s$ making decision $i$. The full list of control variables used consists of the subject's gender, risk preferences elicited using the Holt-Laury multiple price list risk-elicitation procedure, and an indicator variable for (lack of) altruism which is equal to 1 if the individual allocated all points to him(her)self in the dictator game and 0 otherwise. $\beta_{0}$ and $\epsilon$ are the constant and residual terms respectively.

\subsection{Equilibrium predictions} \label{predictions}

Table \ref{table:hybridtable} presents our main regression analysis to test the hypotheses in section \ref{hyp}. Hypothesis 1 (H1) states that participants are less likely to delete links in RPE networks compared to non-RPE ones. Specification (1) regresses the decision to delete links on the ``Network is RPE''. The coefficient is significantly negative ($p<0.001$) validating the hypothesis that a subject is less likely to delete a link in an RPE network than in a non-RPE network. Thus, RPE networks are more likely to be sustained than non-RPE networks, as predicted by the theory.

Columns (2) and (3) replicate the analysis by splitting the RPE variable into SQ and LCC networks respectively. Hypotheses 2 and 3, however, focus on a comparison between different types of RPE networks. In order to address these hypotheses directly, we also add specifications (4)-(7) to restrict the data to decisions on RPE networks only. Moreover, as Figure \ref{allnetworks} shows, a confounding factor in the analysis is that the 1R3 network is both SQ and LCC. In order to isolate its effect, columns (5) and (7) introduce a dummy variable to isolate the effect of 1R3 compared to other SQ or LCC networks. We discuss each hypothesis in turn.

For H2, column (2) in Table \ref{table:hybridtable} splits RPE networks into ``RPE \& non-SQ'' and ``SQ''. The coefficient is significantly negative for both these variables ($p<0.001$) indicating that subjects are less likely to delete links in both types of RPE networks, which is again consistent with H1. Once we look at the size of the coefficients, however, it is clear that the SQ one is not larger than RPE \& non-SQ so these results do not validate the statement in H2 that subjects are less likely to delete links in an SQ network compared to an RPE network that is not SQ.

Column (4) directly shows that subjects are not less likely to delete links in SQ compared to other types of RPE networks. Moreover, once we isolate the effect of 1R3 in column (5), subjects are actually significantly \emph{more} likely to delete links in SQ networks compared to other types of RPE networks ($p<0.001$). The data, therefore, clearly does not support Hypothesis 2 that social quilts is a refinement that identifies a set of networks that sustain cooperation.

\begin{sidewaystable}
\centering
\begin{threeparttable}
\caption{Logit regressions of “Delete” at the decision-level}
\centering
\arrayrulecolor{black}
\label{table:hybridtable}
\begin{tabular}{lccccccc} 
\hline
Sample:               & \multicolumn{3}{c}{All
  networks in Phase 1 and Phase 2} & \multicolumn{4}{c}{~All RPE networks in Phase 1 and Phase 2}                       \\
Baseline:             & \multicolumn{3}{c}{Non-RPE
  networks}                    & \multicolumn{2}{c}{\{2R4,
  1R4, 1R5\}} & \multicolumn{2}{c}{\{2R4,
  2R3, 3R3\}}  \\
\multicolumn{1}{c}{~} & (1)       & (2)       & (3)                               & (4)      & (5)                          & (6)       & (7)                          \\ 
\hline
RPE                   & -0.904*** & ~         & ~                                 & ~        & ~                            & ~         & ~                            \\
~                     & (0.144)   & ~         & ~                                 & ~        & ~                            & ~         & ~                            \\
SQ                    & ~         & -0.787*** & ~                                 & 0.391    & ~                            & ~         & ~                            \\
~                     & ~         & (0.173)   & ~                                 & (0.242)  & ~                            & ~         & ~                            \\
RPE \& non-SQ           & ~         & -1.120*** &                                   & ~        & ~                            & ~         & ~                            \\
~                     & ~         & (0.201)   &                                   & ~        & ~                            & ~         & ~                            \\
LCC                    & ~         & ~         & -1.206***                         & ~        & ~                            & -1.039*** & ~                            \\
~                     & ~         & ~         & (0.159)                           & ~        & ~                            & (0.259)   & ~                            \\
RPE \& non-LCC           & ~         & ~         & -0.131                            & ~        & ~                            & ~         & ~                            \\
~                     & ~         & ~         & (0.249)                           & ~        & ~                            & ~         & ~                            \\
1R3                   & ~         & ~         &                                   & ~        & 0.005                      & ~         & -0.993***                    \\
~                     & ~         & ~         & ~                                 & ~        & (0.257)                      & ~         & (0.269)                      \\
SQ \& not 1R3           & ~         & ~         &                                   & ~        & 1.134***                     & ~         &   ~                  \\
~                     & ~         & ~         & ~                                 & ~        & (0.330)                      & ~         &    ~                   \\
LCC \& not 1R3           & ~         & ~         &                                   & ~        & ~                            & ~         & -1.103***                            \\
~                     & ~         & ~         &     ~                           & ~        & ~                            & ~         & (0.321)                            \\
$\beta_{0}$                 & 0.578**   & 0.554**   & 0.605**                           & -1.145** & -1.045*                      & 0.032    & 0.007                      \\
~                     & (0.243)   & (0.245)   & (0.246)                           & (0.556)  & (0.558)                      & (0.514)   & (0.520)                      \\
~                     & ~         & ~         & ~                                 & ~        & ~                            & ~         & ~                            \\
Observations          & 1,501     & 1,501     & 1,501                             & 367      & 367                          & 367       & 367                          \\
Other Controls        & Yes       & Yes       & Yes                               & Yes      & Yes                          & Yes       & Yes                          \\
Pseudo R2             & 0.0329    & 0.0341    & 0.0444                            & 0.0228   & 0.0524                       & 0.0528    & 0.0532                       \\
Likelihood            & -872.1    & -871.0    & -861.8                            & -245.9   & -238.4                       & -238.3    & -238.2                       \\
\hline
\end{tabular}
\begin{tablenotes}
\footnotesize
\item Note: Logit regressions of “Delete” using decision-level data. Each observation is a decision made by a subject, where “Delete” = 1 if the subject’s decision is to delete a link. Robust standard errors are clustered at the subject-level.
\end{tablenotes}
\arrayrulecolor{black}
\end{threeparttable}
\end{sidewaystable}

For H3, column (3) splits RPE networks into ``LCC'' and ``RPE \& non-LCC''. Interestingly, the coefficient of ``RPE \& non-LCC'' networks is now insignificant, which shows that subjects are not less likely to delete links in RPE compared to non-RPE networks once we exclude LCC networks from the RPE set. In contrast, the LCC coefficient is significantly negative ($p<0.001$) validating the statement in H3 that subjects are less likely to delete links in LCC compared to other types of non-RPE networks.

Column (6) directly shows that subjects are less likely to delete links in LCC compared to other types of RPE networks ($p<0.001$). Moreover, the effect survives even if we separate the impact of 1R3 networks. Both the ``1R3'' and ``LCC \& not 1R3'' variables are significant in specification (7) ($p<0.001$) indicating that cooperation is higher in LCC networks compared to other types of RPE networks even after we isolate the impact of the only LCC network that is also SQ in our experiment. The refinement of the set of RPE networks based on cognitive complexity has behavioral validity -- it identifies a set of network where cooperation is higher because it is easier for participants to understand the consequences of failing to cooperate.

\subsection{Social preferences and learning} \label{social}

The complexity of understanding the consequences of failing to cooperate may mean that subjects need time to learn about equilibrium play. A benefit of the Computer Task between Phases 1 and 2 is to facilitate this learning process by exposing participants to decisions on the same networks they play against other participants, and knowing that the decisions made by the computer players are optimal. In this section we analyse the participants' decisions for Phases 1 and 2 separately in order to investigate the importance of learning in determining equilibrium play, and whether the role of learning varies across types of equilibria and/or networks structures.

Table \ref{table:R2NL} conducts a similar analysis to Table \ref{table:hybridtable} stratifying the data into Phases 1 and 2. The columns labelled with odd numbers contain the analysis using Phase 1 data only, while the columns labelled with even numbers contain the analysis using Phase 2 data only.

Looking just at Phase 1 data, columns (1), (3), and (5) show that cooperation is not more likely in either RPE or SQ networks compared to non-RPE networks. Similarly, (7) shows that cooperation is only marginally more likely in LCC compared to non-RPE networks ($p=0.09$), and (9) shows that the significance is driven by marginally higher cooperation in 1R3 rather than other types of LCC networks ($p=0.07$). In other words, network structure is poorly correlated with cooperation in Phase 1. A significant predictor of cooperative play is, instead, social preferences. In all odd columns with Phase 1 data only, self-interested individuals who allocated all points of the Dictator Game to themselves are more likely to delete links ($p<0.05$). When participants are first exposed to the game, they tend to make decisions driven by their social preferences rather than equilibrium considerations.


\begin{sidewaystable}
\captionsetup{font=large}
\resizebox{\textwidth}{!}{%
\centering
\begin{threeparttable}
\caption{Logit regressions of “Delete” at the decision-level}
\centering
\label{table:R2NL}
\arrayrulecolor{black}
\begin{tabular}{lcccccccccc}
\hline
Sample:                     & \multicolumn{10}{c}{All networks in either Phase 1 or Phase
  2}                                                  \\
Baseline:                   & \multicolumn{10}{c}{Non-RPE networks}                                                                             \\
~                           & Phase 1  & Phase 2   & Phase 1  & Phase 2   & Phase 1  & Phase 2   & Phase 1  & Phase 2   & Phase 1  & Phase 2    \\ 
~                           & (1)  & (2)   & (3)  & (4)   & (5)  & (6)   & (7)  & (8)   & (9)  & (10)    \\ 
\hline
~                           & ~        & ~         & ~        & ~         & ~        & ~         & ~        & ~         & ~        & ~          \\
RPE            & -0.135  & -1.285*** & ~       & ~         & ~       & ~         & ~       & ~         & ~       & ~          \\
~                           & (0.232) & (0.164)   & ~       & ~         & ~       & ~         & ~       & ~         & ~       & ~          \\
SQ             & ~       & ~         & -0.194  & -1.103*** & ~       & ~         & ~       & ~         & ~       & ~          \\
~                           & ~       & ~         & (0.276) & (0.197)   & ~       & ~         & ~       & ~         & ~       & ~          \\
RPE \& non-SQ    & ~       & ~         & -0.003  & -1.608*** & -0.001  & -1.607*** & ~       & ~         & ~       & ~          \\
~                           & ~       & ~         & (0.376) & (0.233)   & (0.376) & (0.233)   & ~       & ~         & ~       & ~          \\
LCC             & ~       & ~         & ~       & ~         & ~       & ~         & -0.434* & -1.612*** & ~       & ~          \\
~                           & ~       & ~         & ~       & ~         & ~       & ~         & (0.256) & (0.182)   & ~       & ~          \\
RPE \& non-LCC    & ~       & ~         & ~       & ~         & ~       & ~         & 0.688   & -0.474*   & 0.689   & -0.475*    \\
~                           & ~       & ~         & ~       & ~         & ~       & ~         & (0.469) & (0.281)   & (0.469) & (0.281)    \\
1R3            & ~       & ~         & ~       & ~         & -0.568* & -1.534*** & ~       & ~         & -0.568* & -1.535***  \\
~                           & ~       & ~         & ~       & ~         & (0.317) & (0.231)   & ~       & ~         & (0.317) & (0.231)    \\
SQ \& not 1R3      & ~       & ~         & ~       & ~         & 0.507   & -0.234    & ~       & ~         & ~       & ~          \\
~                           & ~       & ~         & ~       & ~         & (0.480) & (0.314)   & ~       & ~         & ~       & ~          \\
LCC \& not 1R3    & ~       & ~         & ~       & ~         & ~       & ~         & ~       & ~         & -0.221  & -1.721***  \\
~                           & ~       & ~         & ~       & ~         & ~       & ~         & ~       & ~         & (0.391) & (0.252)    \\
Subject
  is purely selfish & 0.453** & -0.008    & 0.449** & 0.005     & 0.471** & -0.033    & 0.475** & -0.038    & 0.473** & -0.032     \\
~                           & (0.190) & (0.185)   & (0.190) & (0.185)   & (0.188) & (0.188)   & (0.188) & (0.187)   & (0.188) & (0.187)    \\
$\beta_{0}$ & 0.524 & 0.661** & 0.535 & 0.622** & 0.541 & 0.683** & 0.527 & 0.710*** & 0.542 & 0.696***\\
~                           & (0.350) & (0.305) & (0.352) & (0.307) & (0.354) & (0.314) & (0.352) & (0.308) & (0.353) & (0.310)    \\
~                           & ~       & ~         & ~       & ~         & ~       & ~         & ~       & ~         & ~       & ~          \\
Observations                & 599     & 902       & 599     & 902       & 599     & 902       & 599     & 902       & 599     & 902        \\
Other
  Controls            & Yes     & Yes       & Yes     & Yes       & Yes     & Yes       & Yes     & Yes       & Yes     & Yes        \\
Pseudo
  R2                 & 0.0098  & 0.0657    & 0.0101  & 0.0689    & 0.0162  & 0.0810    & 0.0179  & 0.0796    & 0.0187  & 0.0799     \\
Likelihood                  & -344.2  & -516.7    & -344.1  & -514.9    & -342.0  & -508.2    & -341.4  & -509.0    & -341.2  & -508.8     \\
\hline
\end{tabular}
\begin{tablenotes}
\footnotesize
\item Note: Logit regressions of “Delete” using decision-level data. Each observation is a decision made by a subject, where “Delete” = 1 if the subject’s decision is to delete a link. Robust standard errors are clustered at the subject-level.
\end{tablenotes}
\arrayrulecolor{black}
\end{threeparttable}
}
\end{sidewaystable}


In Phase 2, after participants have had the opportunity to learn from the Computer Task, however, equilibrium considerations become a significant predictor of decisions to cooperate. Specification (2) shows that, as stated in H1, participants are less likely to delete links in RPE compared to non-RPE networks ($p<0.001$). If we limit our attention to SQ networks, column (4) shows that participants are less likely to delete links in SQ compared to non-RPE networks ($p<0.001$). However, column (6) shows that this effect is driven by the 1R3 network that is both SQ and LCC -- participants are not more likely to cooperate when assigned to other types of SQ networks compared to non-RPE networks ($p=0.44$), and therefore we find little evidence to support H2.

Given that the motivation to formulate the concept of LCC is based on the low burden of computation to understand these equilibria, we would expect equilibrium play to emerge in LCC networks after the Computer Task facilitates learning. Column (8) confirms that in Phase 2 cooperation is higher in LCC networks compared to non-RPE networks ($p<0.001$). Moreover, column (10) shows that cooperation in LCC networks is not just driven by cooperation in the 1R3 network -- the LCC coefficient is significant even after we exclude 1R3 networks ($p<0.001$). 

Experimental data suggests that LCC is a better behavioral predictor of cooperative play than SQ. The analysis in Table \ref{table:R3NL} delves deeper into the comparison of cooperation in SQ and LCC networks. First, we exclude data from the 2R4 network, which is the only RPE network that is neither SQ nor LCC. Second, specifications (1) and (2) investigate cooperation levels in LCC and SQ compared to non-RPE networks, and they separate the 1R3 network that is both SQ and LCC from either category. As we have seen above, network structure is not significant in Phase 1, but it becomes a predictor of cooperation in Phase 2. In particular, column (2) shows cooperation is significantly higher in both LCC ($p<0.001$) and 1R3 ($p<0.001$) compared to non-RPE networks, but there is no difference between SQ and non-RPE networks ($p=0.45$). A Wald test confirms that there is more cooperation in both LCC ($p<0.001$) and 1R3 ($p<0.001$) compared to SQ networks. 

Finally, specifications (3) and (4) limit the data to RPE networks and take 1R3 as the baseline comparison network. Column (4) with Phase 2 data only shows cooperation in LCC networks other than 1R3 is indistinguishable from cooperation in 1R3, while cooperation in 1R3 is higher than SQ networks other than 1R3 ($p<0.001$). Again, a Wald test confirms that there is more cooperation in LCC compared to SQ networks ($p<0.001$).


\begin{table}
\centering
\begin{threeparttable}
\caption{Logit regressions of “Delete” at the decision-level}
\centering
\label{table:R3NL}
\arrayrulecolor{black}
\begin{tabular}{lcccc} 
\hline
Sample:                                        & \multicolumn{2}{c}{All networks excl. 2R4} & \multicolumn{2}{c}{RPE networks excl. 2R4}  \\
Baseline:                                      & \multicolumn{2}{c}{Non-RPE networks}       & \multicolumn{2}{c}{1R3 networks}                  \\
                                               & Phase 1 & Phase 2                          & Phase 1 & Phase 2                                 \\ 
                                                                                              & (1) & (2)                          & (3) & (4)                                 \\ 
\arrayrulecolor{black}\cline{1-1}\arrayrulecolor{black}\cline{2-5}
\textit{Logit regressions on “Delete”} & ~       & ~                                & ~       & ~                                       \\
\textit{~}                                     & ~       & ~                                & ~       & ~                                       \\
\textit{$\beta_{LCC}$}: LCC \& not 1R3               & -0.220  & -1.723***                        & 0.410   & -0.319                                  \\
~                                                          & (0.391) & (0.252)                          & (0.486) & (0.321)                                 \\
\textit{$\beta_{SQ}$}: SQ \& not 1R3               & 0.506   & -0.238                           & 1.096** & 1.133***                                \\
~                                                          & (0.484) & (0.314)                          & (0.528) & (0.348)                                 \\
\textit{$\beta_{1R3}$}: 1R3                     & 0.569*  & -1.538***                        & ~       & ~                                       \\
~                                                          & (0.316) & (0.232)                          & ~       & ~                                       \\
~                                                          & ~       & ~                                & ~       & ~                                       \\
Observations                                               & 593     & 885                              & 114     & 230                                     \\
Other
  Controls                                           & Yes     & Yes                              & Yes     & Yes                                     \\
Pesudo
  R2                                                & 0.0161  & 0.0820                           & 0.0332  & 0.0790                                  \\
Likelihood                                                 & -340.2  & -495.9                           & -68.0   & -146.8                                  \\ 
\hline
\textit{Wald tests}                                        & ~       & ~                                & ~       & ~                                       \\
\textit{~}                                                 & ~       & ~                                & ~       & ~                                       \\
$\beta_{LCC}-\beta_{SQ}$                                    & -0.726  & -1.485***                        & -0.686  & -1.452***                               \\
~                                                          & (0.605) & (0.393)                          & (0.616) & (0.402)                                 \\
$\beta_{LCC}-\beta_{1R3}$                                   & 0.349   & -0.185                           & ~       & ~                                       \\
~                                                          & (0.482) & (0.319)                          & ~       & ~                                       \\
$\beta_{SQ}-\beta_{1R3}$                                   & 1.074** & 1.300***                         & ~       & ~                                       \\
~                                                          & (0.531) & (0.356)                          & ~       & ~                                       \\
\hline
\end{tabular}
\begin{tablenotes}
\footnotesize
\item Note: Logit regressions of “Delete” using decision-level data. Each observation is a decision made by a subject, where “Delete” = 1 if the subject’s decision is to delete a link. Robust standard errors are clustered at the subject-level. 2R4 networks were removed from the regressions to standardize the baselines. Wald tests for differences in coefficients are also provided in the bottom panel.
\end{tablenotes}
\arrayrulecolor{black}
\end{threeparttable}
\end{table}

The Computer Task also allows us to disentangle social preferences and/or beliefs about other participants' cognitive ability as a driver of decisions because these two factors are not relevant when playing against computer players. Table \ref{table:F} compares participants' play in Phases 1 and 2 with the Computer Task, conditioning on the same networks seen by the same participant. In particular, we compare the difference in the ratio of optimal decisions made by participants in each pair of tasks (e.g. Phase 1 optimality ratio versus Computer Task optimality ratio), stratifying the decisions by ``All networks'', ``RPE networks only'', ``LCC networks only'', and ``SQ networks only''. Here, we define an ``optimal'' decision as a decision that a computer player would have made (i.e., delete a link if the network is not RPE, not delete a link if the network is RPE). The ``Delete Ratio'' of a subject in a task represents the ratio of number of delete decisions to total number of decisions of that subject in that task.

Participants play less optimally in Phase 1 compared to the Computer Task. The negative sign in column (1) shows that participants' decisions in Phase 1 are significantly less optimal on all networks, while the positive signs in (2)-(4) show participants remove more links in RPE networks -- a decision that is not optimal -- in Phase 1 compared to the Computer Task. Once we compare play in Phase 2 and the Computer Task, however, none of the differences are significant. While qualitatively it still holds that play tends to be more optimal in the Computer Task, the lack of significance suggests that social preferences and/or beliefs about other participants' limited cognitive ability are second-order considerations in determining cooperative play.

\clearpage
\begin{table}
\resizebox{\textwidth}{!}{%
\centering
\begin{threeparttable}
\caption{Tests for differences in means between Phase 1 or 2 vs Computer Task delete and optimality ratios at the subject-task-level.}
\label{table:F}
\begin{tabular}{lcccc} 
\hline
Sample:                                      & \multicolumn{1}{l}{All Networks}     & \multicolumn{1}{l}{RPE networks} & \multicolumn{1}{l}{SQ Networks} & \multicolumn{1}{l}{LCC Networks}  \\
Ratio:                                       & \multicolumn{1}{l}{Optimality Ratio} & \multicolumn{1}{l}{Delete Ratio}      & \multicolumn{1}{l}{Delete Ratio}     & \multicolumn{1}{l}{Delete Ratio}      \\ 
~                                            & (1)                                 & (2)                                  & (3)                                 & (4)                                  \\
\hline
\textit{Wilcoxon Signed-rank tests} & \multicolumn{1}{l}{~}                & \multicolumn{1}{l}{~}                 & \multicolumn{1}{l}{~}                & \multicolumn{1}{l}{~}                 \\
~                                            & \multicolumn{1}{l}{~}                & \multicolumn{1}{l}{~}                 & \multicolumn{1}{l}{~}                & \multicolumn{1}{l}{~}                 \\
Phase 1 vs Computer Task                           & -3.538***                            & 3.704***                              & 2.865***                             & 3.715***                              \\
~                                            & [82]                                 & [78]                                  & [47]                                 & [76]                                  \\
Phase 2 vs Computer Task                         & -1.006                               & 1.346                                 & 1.395                                & 1.313                                 \\
~                                            & [171]                                & [161]                                 & [98]                                & [160]                                  \\
~                                            & \multicolumn{1}{l}{~}                & \multicolumn{1}{l}{~}                 & \multicolumn{1}{l}{~}                & \multicolumn{1}{l}{~}                 \\
\hline
\end{tabular}
\begin{tablenotes}
\small
\item Note: The square brackets denote the number of observations. Each observation is a subject in a task (i.e., each subject contributes to some number $x\in\{0,1,2,3\}$ of observations corresponding to Phase 1/Phase 2). We test the hypothesis that the mean difference of the optimality ratios is less than zero, and the hypothesis that the mean difference of the delete ratios is more than zero. We report the z-score for the Wilcoxon signed-rank tests.
\end{tablenotes}
\arrayrulecolor{black}
\end{threeparttable}}
\end{table}

\section{Conclusion} \label{conclusion}

This paper has shown both theoretically and experimentally the importance of bounded rationality in mediating how social network structure determines cooperative behavior. In the theory part, we defined a novel notion of \emph{Low Cognitive Complexity} (LCC) to select among cooperative (RPE) equilibria in the \cite{jackson:2012} (JRT) framework. Our two main contributions come from the experiment. First, we validate the core prediction of the JRT model that RPE equilibria are more cooperative that non-equilibrium network structures. Second, we show that cooperation is higher in LCC compared to other types of equilibrium (RPE) networks.

Our results point to the importance of further research at the intersection of bounded rationality and networks. The intrinsic complexity of network structure means that equilibrium predictions based on full rationality and complete information about the network are unlikely to be descriptive of actual behavior. Iterative thinking is an appealing way to deal with this complexity in the network context because, loosely speaking, individuals tend to think about the consequences of their actions up to some steps away from their network position, and the ``some'' can be easily related to the steps of iterative reasoning the individual is able to perform. This paper shows that bounded rationality based on iterative reasoning helps to select among a multiplicity of equilibria to predict how network structure affects cooperative behavior. We believe this approach may help in resolving the problem of equilibrium multiplicity that is present in many other contexts where networks affect behavior.

Our results provide further evidence that the ability to remove links affects behavior. Given that most theoretical contributions take the network as exogenous and unchangeable, previous experimental research incorporating network formation tended to be atheoretical with a lot of variation in experimental protocols in terms of how links can be added/removed by participants. In our paper, instead, we generate precise theoretical predictions by adopting a reduced form of the network formation process where only link removal is allowed in order to match the JRT theoretical framework. A fruitful avenue for further research is applying this reduced form approach to other games on networks to investigate theoretically and experimentally how the inclusion of the network formation process matters in determining how networks affect our behavior.


\newpage
\bibliography{refs}

\newpage
\section*{Appendix A} \label{appendixa}

In this section, we prove the novel theoretical results in the paper and discuss their validity beyond the $2b>c>b$ parametrization adopted in the paper.

Let $m$ be a positive integer such that $mb > c > (m-1)b$. We focus on the case where $m \geq 2$ since $m = 1$ is a trivial case where all networks are RPE. Note that when $m=2$ this is the parametrization in our paper. Definition \ref{defn:tc} can then be generalised as:

\begin{defn}
\label{defn:tc-general}
The class of \textbf{TC networks} are generated as follows: For $k=0$, define $TC_0(m)=\{g^{\emptyset}\}$. Inductively on $k$, $TC_k(m)$ is such that $g \in TC_k(m)$ if and only if $\forall i,ij \subset g, \exists g' \subset g-ij $ such that 
\begin{itemize}
 \item $g' \in TC_{k'}(m)$ for some $k'\leq k-m$; 
 \item $d_i (g') \leq d_i (g)-m$
 \item $ \nexists g'' \in TC_{k''} $ s.t. $ g'' \subset g-ij $ and $ D(g'')>D(g') $.
\end{itemize}
\end{defn}

Intuitively, this is the same as Definition \ref{defn:tc} except that one would need to lose $m$ links as a result of deleting a link, in order to discourage the agent from deleting the link. JRT show that Proposition \ref{prop:RPE} and \ref{prop:quilt} extend to this set-up with an appropriate generalization of the definition of social quilt.

Below we prove Proposition \ref{prop:lowCC} for the $m=2$ case and discuss its extension to the general case when $m>2$. We start by proving three lemmas that apply to the general $m\geq 2$ case.

\begin{lem}
\label{lem:lowCCproof-1}
All nodes in RPE networks either have no links or have $m$ or more links.
\end{lem}

\begin{proof}
\noindent
Suppose that a node in a RPE network only has $l$ links, such that $l \leq (m-1)$. Given that $c>(m-1)b \geq lb$, the agent can get a higher payoff by refusing a favour, since the maximum possible value of retaining all links is $lb$. This network therefore cannot be RPE, and we have reached a contradiction.
\end{proof}

\begin{lem}
\label{lem:lowCCproof-2}
The number of links in RPE networks are more than or equal to the number of nodes with links.
\end{lem}

\begin{proof}
\noindent
We limit the graph only to the nodes with links, $g_e \subseteq g$. Let the number of nodes in $g_e$ be $N_e$ and the number links be $|E(g_e )|$.

From Proposition \ref{lem:lowCCproof-1}, $d_i (g_e ) \geq m$ $\forall$ $i \in g_e$. This implies that $\sum_{i \in g_e} d_i (g_e ) \geq m \times N_e$. 

Given that, by definition, $2|E(g_e )|=\sum_{i \in g_e} d_i (g_e )$, this implies that $2|E(g_e )| \geq m \times N_e$ and therefore $|E(g_e )|\geq N_e$ since $m \geq 2$.
\end{proof}

\begin{lem}
\label{lem:lowCCproof-3}
All non-empty RPE networks contain a cycle.
\end{lem}

\begin{proof}
\noindent
A graph with $n$ vertices is a tree if and only if it is connected and has $n-1$ links. This means that a graph with $n$ vertices and at least $n$ links contains a cycle. 

From Lemma \ref{lem:lowCCproof-2}, this implies that all non-empty RPE networks, of which $g_e$ is a subset, contain a cycle. 
\end{proof}

The above lemmas 
apply for the general case where $m \geq 2$. However, the following lemma is only valid if $m = 2$.

\begin{lem}
\label{lem:lowCCproof-4}
For $m=2$, a simple cycle is a RPE network.
\end{lem}

\begin{proof}
\noindent
For $m=2$, in a simple cycle, any agent who deletes a link will cause the agent with the deleted link to only have one link remaining. This agent will thus delete the remaining link he has, and this process of deleting continues until we reach an empty network. All agents would have lost two links in the resulting equilibrium, and therefore no agent has an incentive to refuse a favour, making this an RPE network.
\end{proof}

Note that Lemma \ref{lem:lowCCproof-4} does not apply if $m>2$, since $d_i(g) = 2$ where $i$ are nodes in the simple cycle $g$, and this would contradict Lemma \ref{lem:lowCCproof-1}.

We now prove the following proposition stated in the paper for $m=2$.

\noindent
\textbf{Proposition \ref{prop:lowCC}.} \textit{For m=2, a network is LCC if and only if it is a simple cycle.}

\begin{proof}
\noindent
$(\Leftarrow)$ In Lemma \ref{lem:lowCCproof-4}, we have shown that any deletion of a link in a simple cycle will cause a cascade of deletion that results in the empty network. Simple cycles are therefore LCC.

$(\Rightarrow)$ We prove this by contradiction. Suppose that there exists a LCC RPE network g that is not a simple cycle. The empty network is excluded. From Lemma \ref{lem:lowCCproof-3}, there exists a simple cycle within the network. Let that simple cycle be $g'$. Since $g$ is not a simple cycle, the simple cycle $g'$ must be a strict-subset of $g$; i.e. $g' \subset g$.

By Lemma \ref{lem:lowCCproof-4}, this simple cycle is RPE; i.e. $g' \in RPE$.
By definition of LCC, the largest strict-subset RPE-network must be the empty network. However, $g' \subset g$, $g' \in RPE$, and $g' \neq \{g^{\emptyset}\}$.

Thus, we have reached a contradiction since we have found a non-empty strict-subset that is an RPE network.
\end{proof}

Since Lemma \ref{lem:lowCCproof-4} does not hold for $m>2$, it follows that Proposition \ref{prop:lowCC} does not hold for $m>2$ since simple cycles are not RPE under $m>2$ in the first instance. To our knowledge, there is no simple defining characteristic for LCC networks under $m>2$, and we would have to use the Definition \ref{defn:tc-general} to characterise it, i.e. a RPE network where the largest TC subnetwork is the empty network. Given that there are no simple defining characteristics for $m>2$, it is likely that participants with bounded rationality will find it difficult to predict equilibrium outcomes even in LCC networks for $m>2$. Hence, while the theoretical results would still apply for $m>2$, caution should be applied in assuming that LCC networks for $m>2$ would be computationally easy for participants to understand the consequences of a defection. Future work could explore experimenting with LCC networks for $m>2$ to test this hypothesis.

\newpage

\section*{Appendix B} \label{appendixb}
This section provides screenshots of the specific instructions and tasks given to participants in the experiment. In Task 1, participants' levels of altruism are measured via the Dictator Game, after which the instructions of the network game are provided. Before they proceed, participants must pass a Quiz which tests their knowledge on how they are being compensated in a random example of decisions made in a given network game. This is to ensure that participants have a basic understanding of the game.

Once participants pass the Quiz, they proceed to play the network game in Tasks 2, 3, 4, and 5. Tasks 2, 4, and 5 are the tasks in which participants play against each other, whereas participants play against computer players in Task 3. In our paper, we refer to Task 2 as Phase 1, Task 3 as the Computer Task, and Tasks 4 and 5 as Phase 2. 

Towards the end of the experiment, we measure participants' ability to perform backward induction and risk aversion using \cite{holt:2002} in Tasks 6 and 7 respectively\footnote{In Task 6, the computer and the subject take turns adding 1 or 2 points to a pot with 0 points at the beginning. The computer adds randomly every turn and makes the first move. The winner is the player that reaches 21 points. Subjects with an inclination for backward induction would realize that in order to achieve a guaranteed win, the pot needs to be at 18, 15, 12, 9, 6 and 3 points during the computer's turn. Unfortunately, we failed to record the moves of the subjects --with the exception of the end result -- due to a code error. Although winning may be determined by the ability to do backward induction, luck can play a significant role. To be conservative, we did not consider subjects' performance in Task 6 in our analysis}. Finally, the experiment concludes with a survey.

\vspace{0.2in}
\begin{figure}[!hbt]
    \centering
    \includegraphics[width=0.8\linewidth]{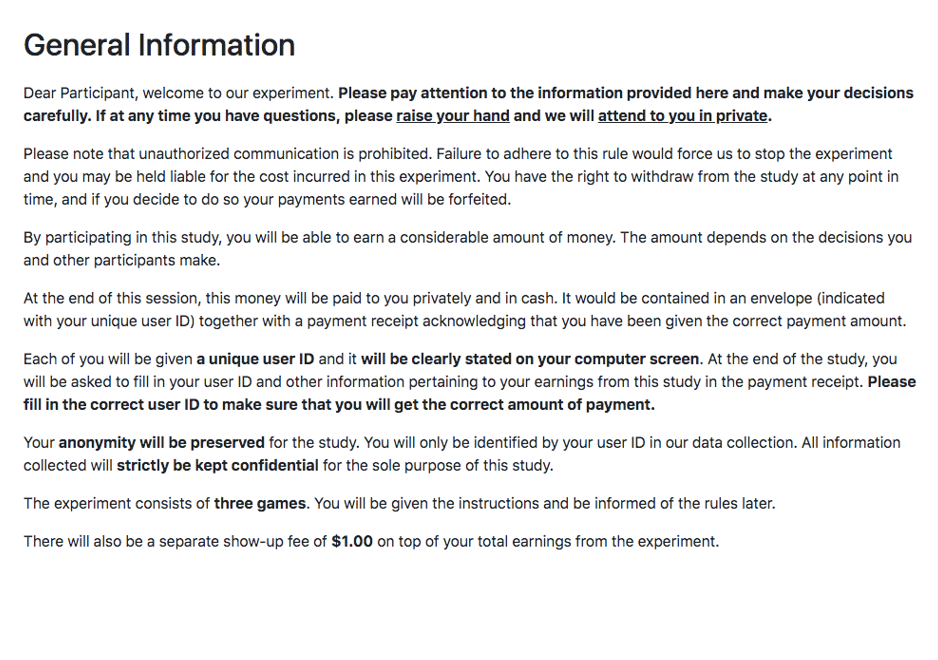}
\end{figure}
\begin{figure}[!hbt]
    \centering
    \includegraphics[width=0.8\linewidth]{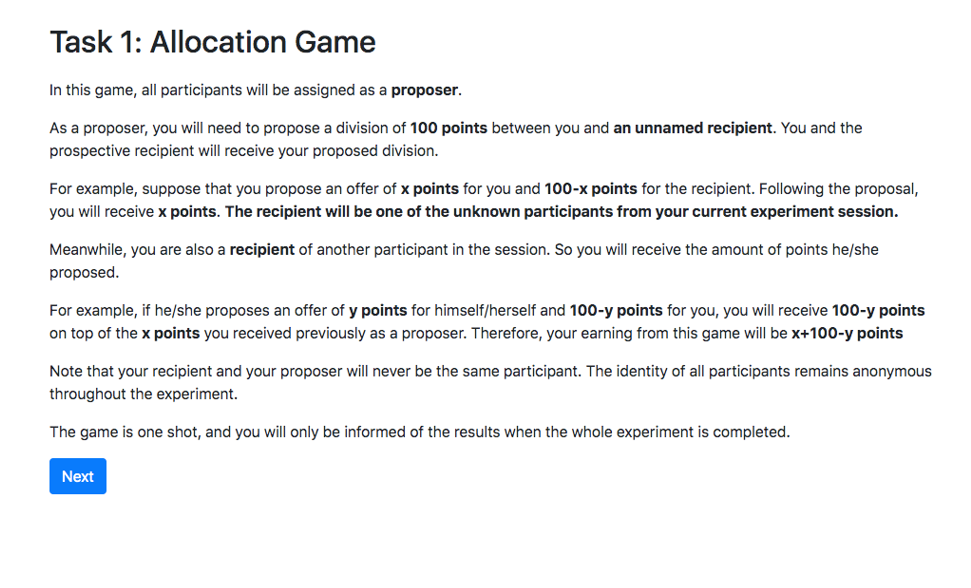}
\end{figure}
\begin{figure}[!hbt]
    \centering
    \includegraphics[width=0.8\linewidth]{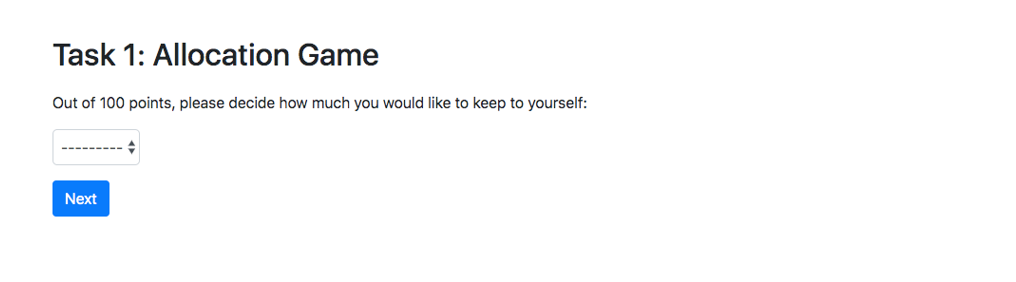}
\end{figure}
\begin{figure}[!hbt]
    \centering
    \includegraphics[width=0.8\linewidth]{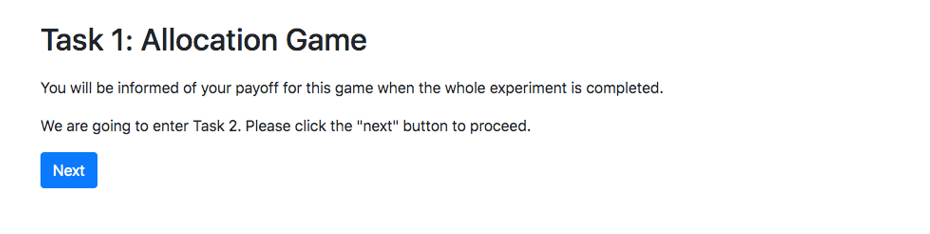}
\end{figure}
\clearpage
\begin{figure}[!hbt]
    \centering
    \includegraphics[width=0.8\linewidth]{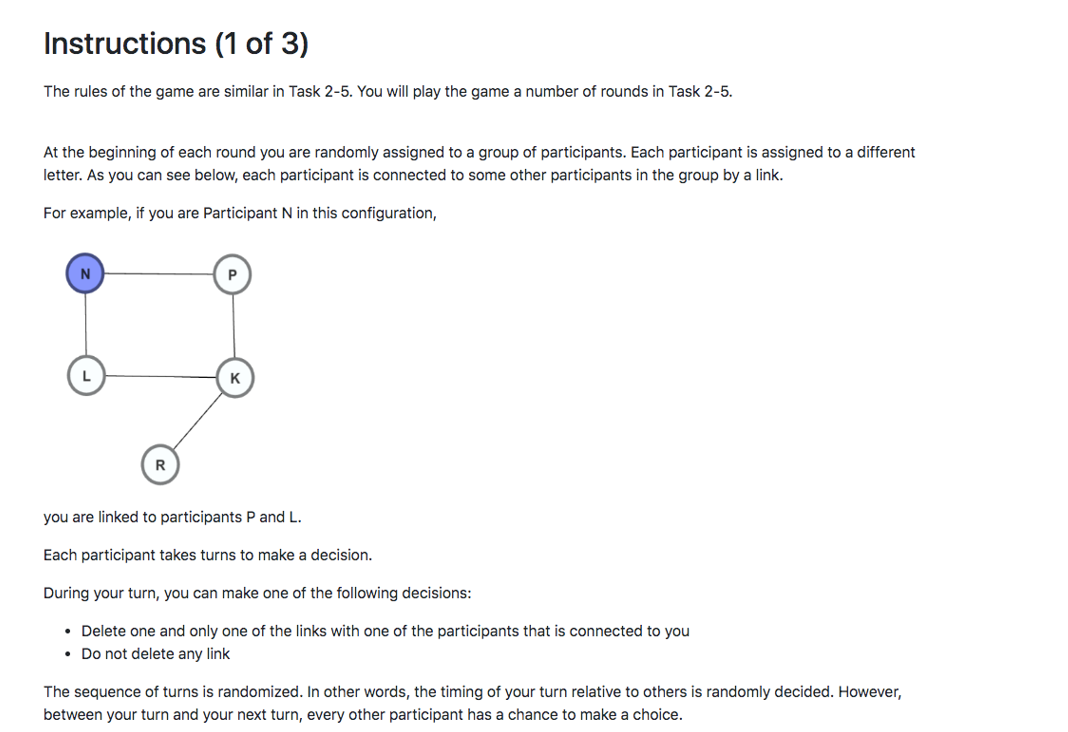}
\end{figure}
\begin{figure}[!hbt]
    \centering
    \includegraphics[width=0.8\linewidth]{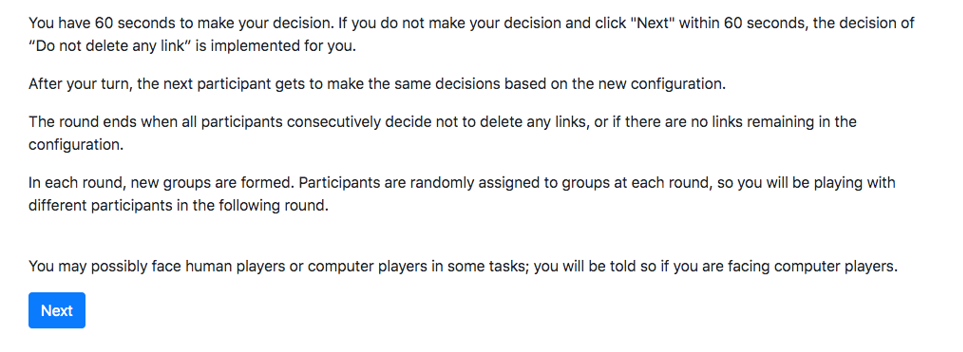}
\end{figure}
\begin{figure}[!hbt]
    \centering
    \includegraphics[width=0.8\linewidth]{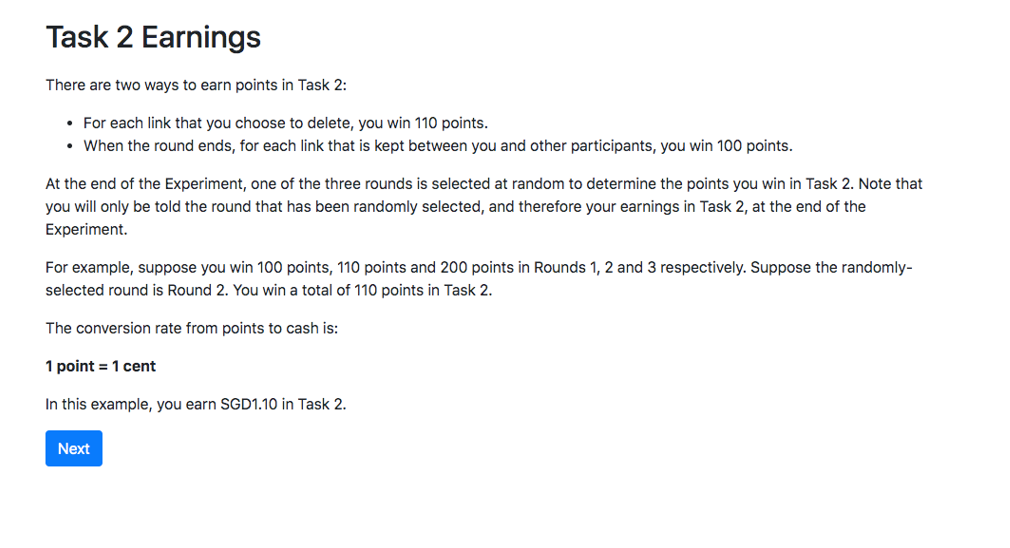}
\end{figure}
\begin{figure}[!hbt]
    \centering
    \includegraphics[width=0.8\linewidth]{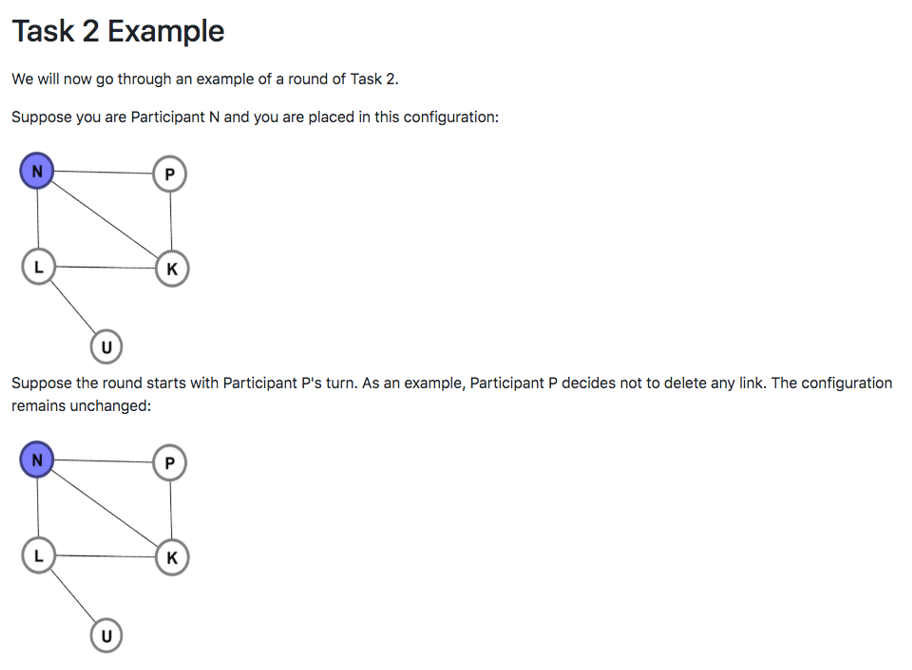}
\end{figure}
\begin{figure}[!hbt]
    \centering
    \includegraphics[width=0.8\linewidth]{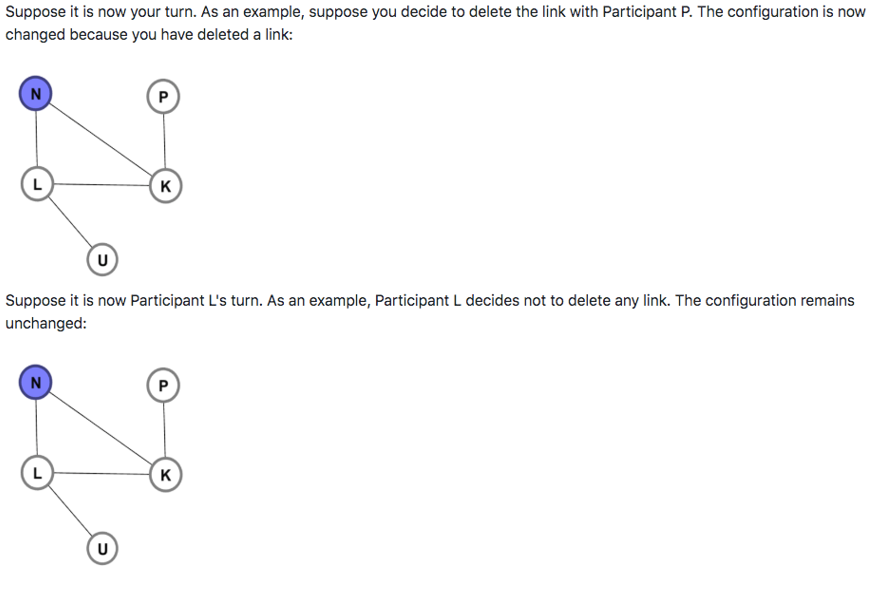}
\end{figure}
\begin{figure}[!hbt]
    \centering
    \includegraphics[width=0.8\linewidth]{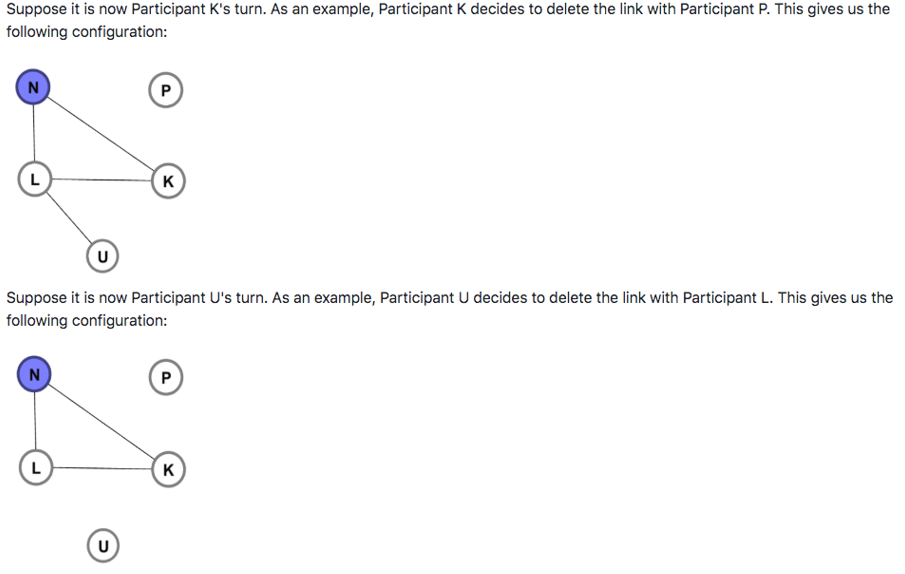}
\end{figure}
\begin{figure}[!hbt]
    \centering
    \includegraphics[width=0.8\linewidth]{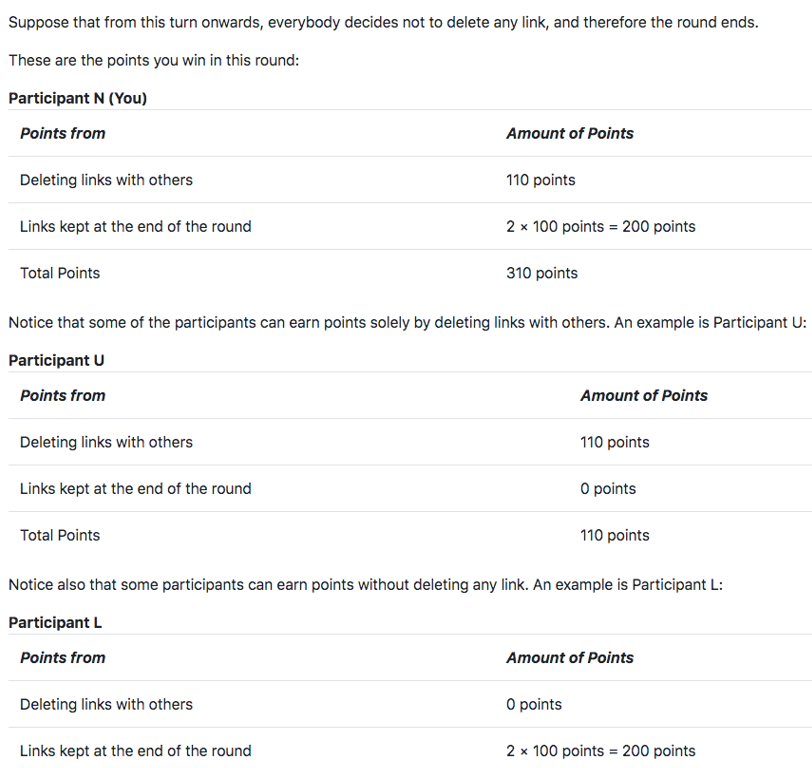}
\end{figure}
\begin{figure}[!hbt]
    \centering
    \includegraphics[width=0.8\linewidth]{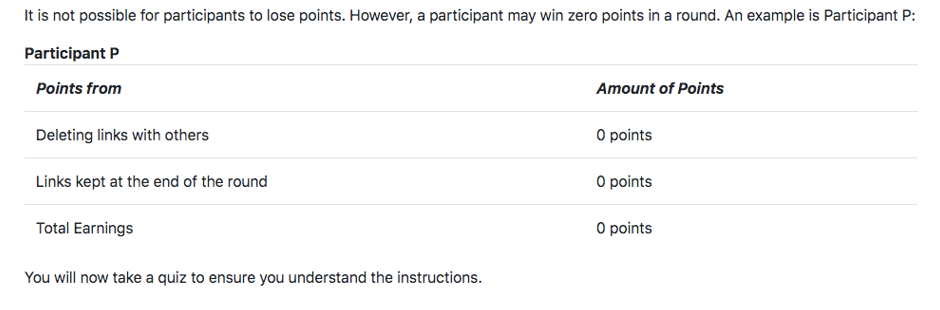}
\end{figure}
\begin{figure}[!hbt]
    \centering
    \includegraphics[width=0.8\linewidth]{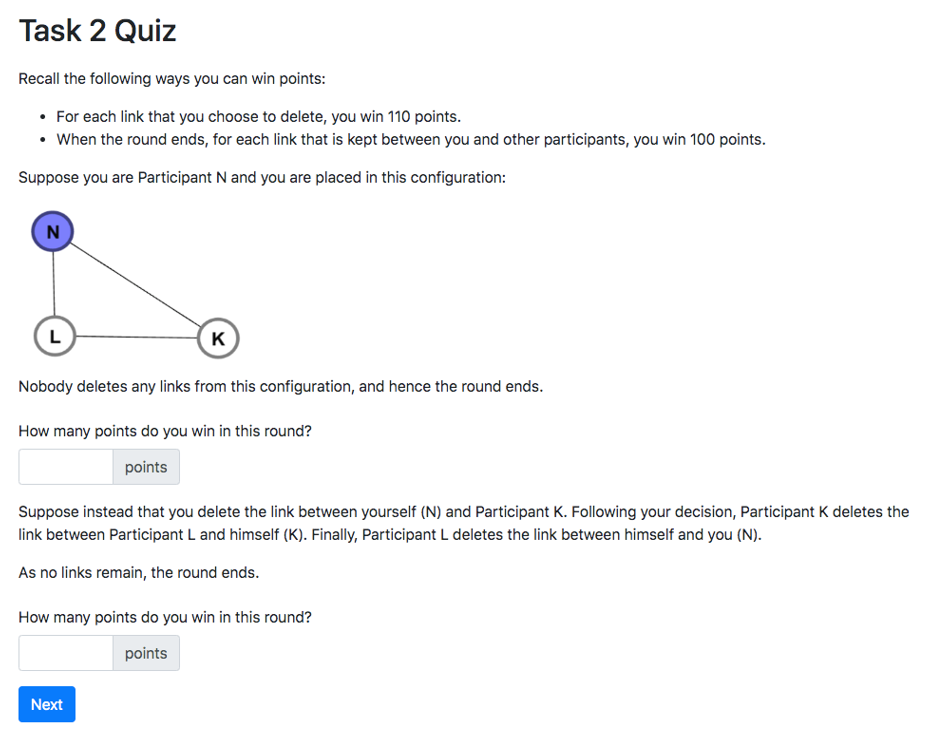}
\end{figure}
\begin{figure}[!hbt]
    \centering
    \includegraphics[width=0.8\linewidth]{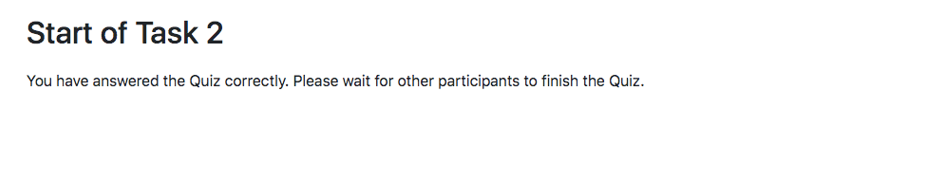}
\end{figure}
\begin{figure}[!hbt]
    \centering
    \includegraphics[width=0.8\linewidth]{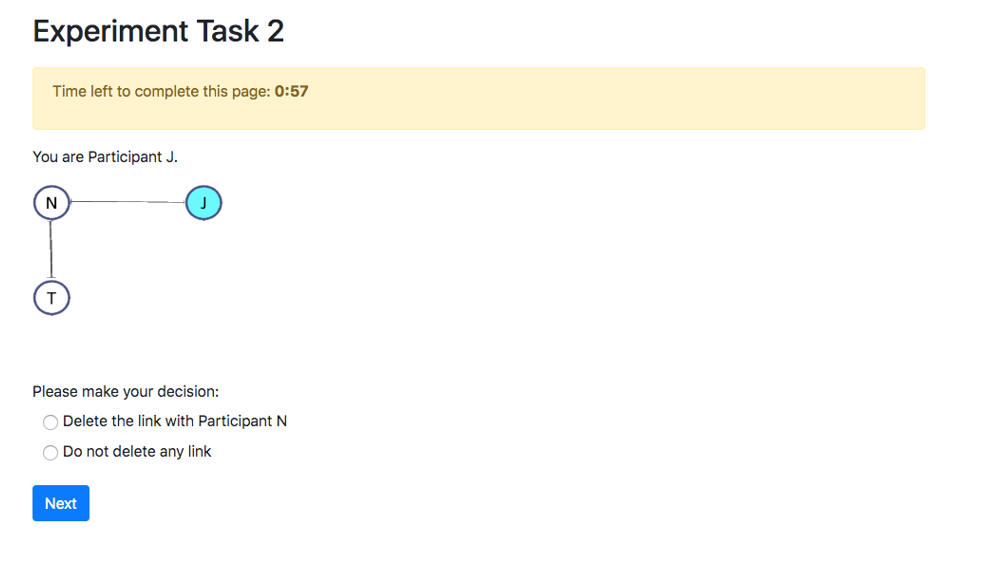}
\end{figure}
\begin{figure}[!hbt]
    \centering
    \includegraphics[width=0.8\linewidth]{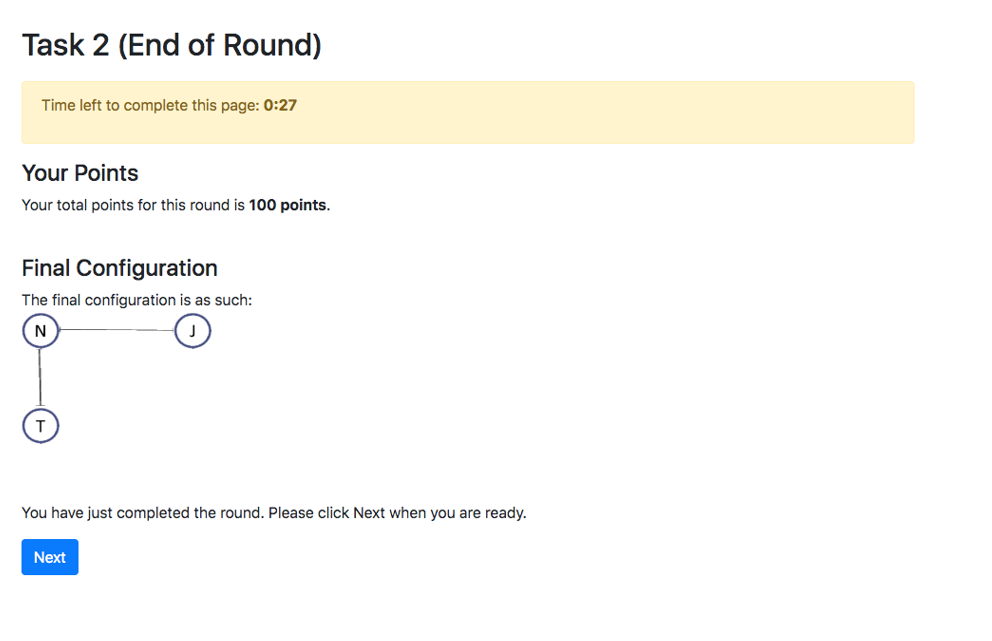}
\end{figure}
\begin{figure}[!hbt]
    \centering
    \includegraphics[width=0.8\linewidth]{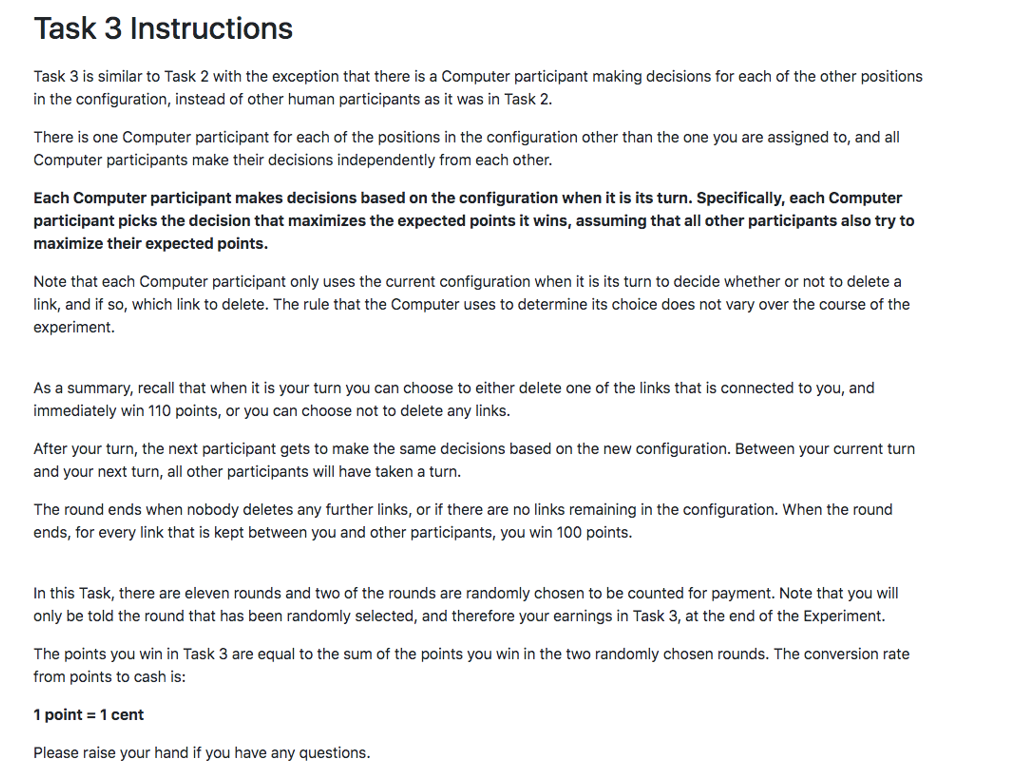}
\end{figure}
\begin{figure}[!hbt]
    \centering
    \includegraphics[width=0.8\linewidth]{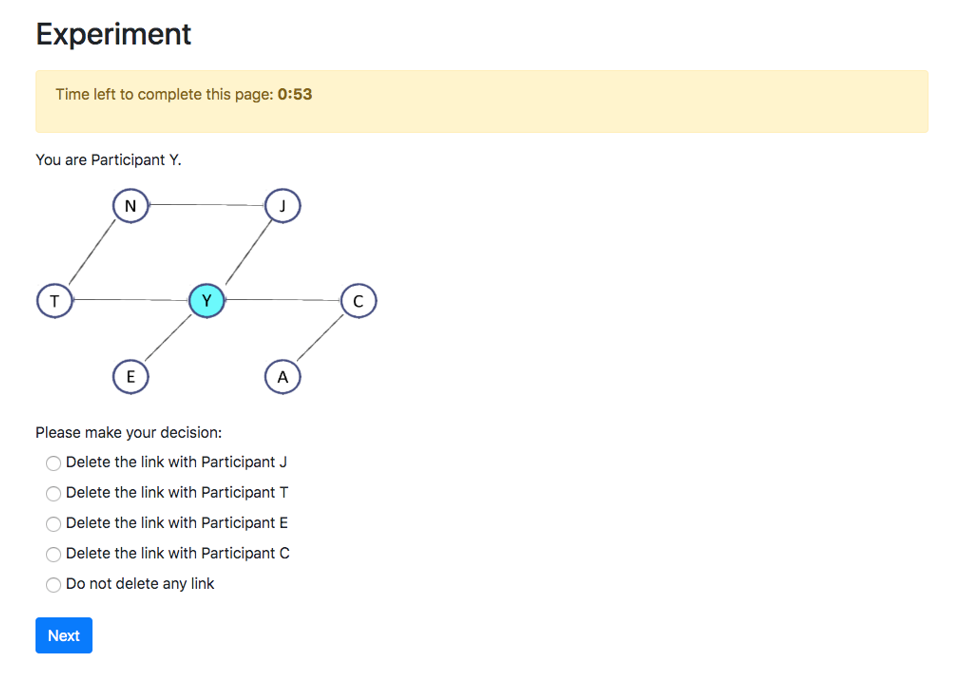}
\end{figure}
\begin{figure}[!hbt]
    \centering
    \includegraphics[width=0.8\linewidth]{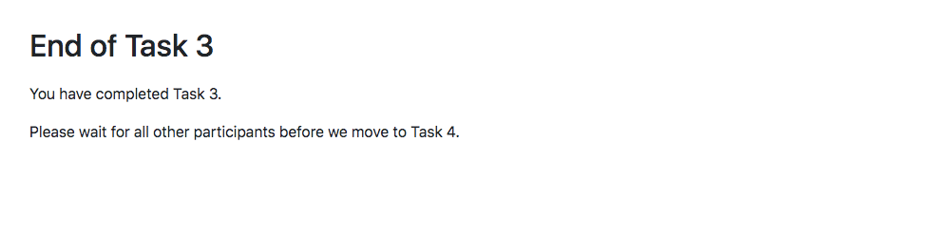}
\end{figure}
\begin{figure}[!hbt]
    \centering
    \includegraphics[width=0.8\linewidth]{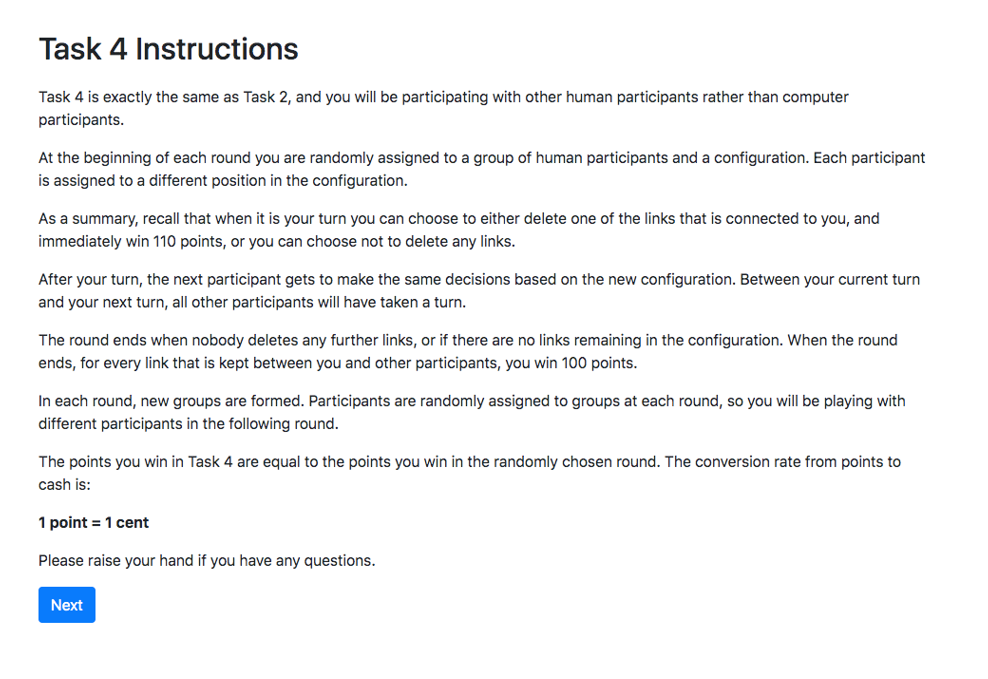}
\end{figure}
\begin{figure}[!hbt]
    \centering
    \includegraphics[width=0.8\linewidth]{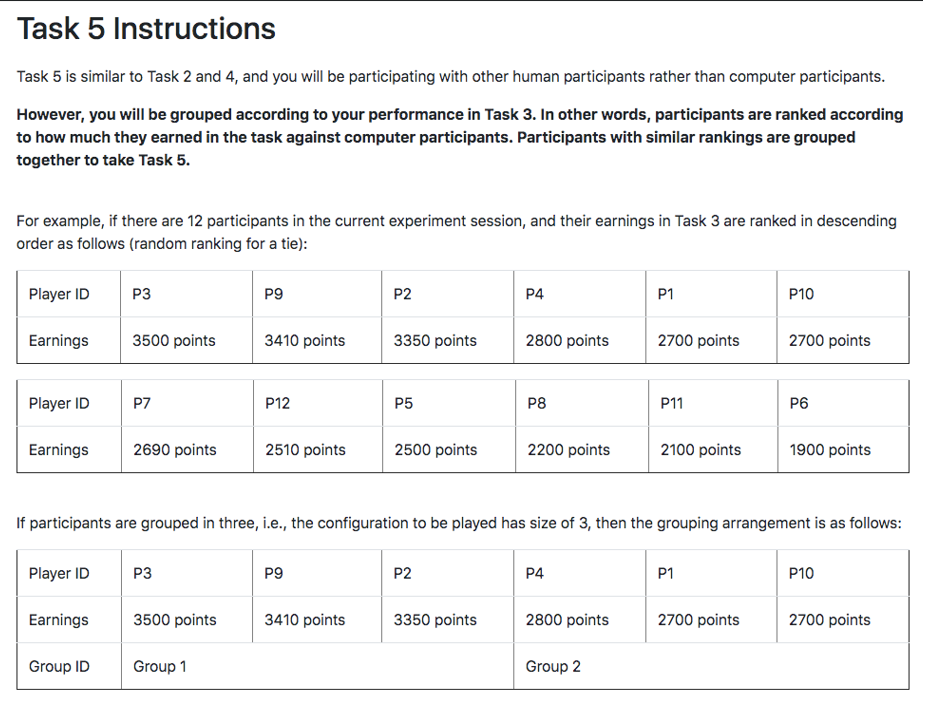}
\end{figure}
\begin{figure}[!hbt]
    \centering
    \includegraphics[width=0.8\linewidth]{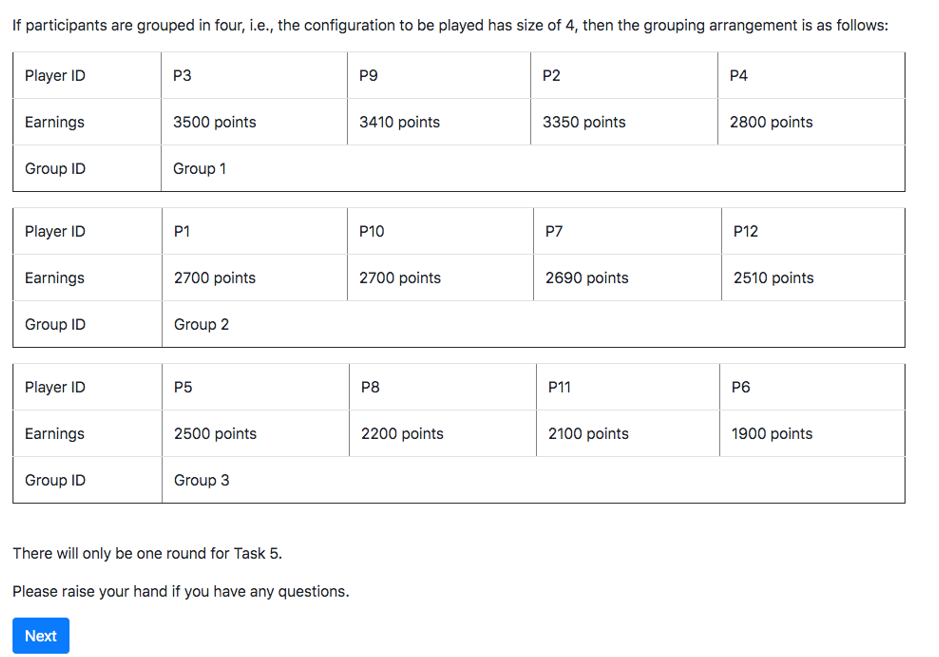}
\end{figure}
\begin{figure}[!hbt]
    \centering
    \includegraphics[width=0.8\linewidth]{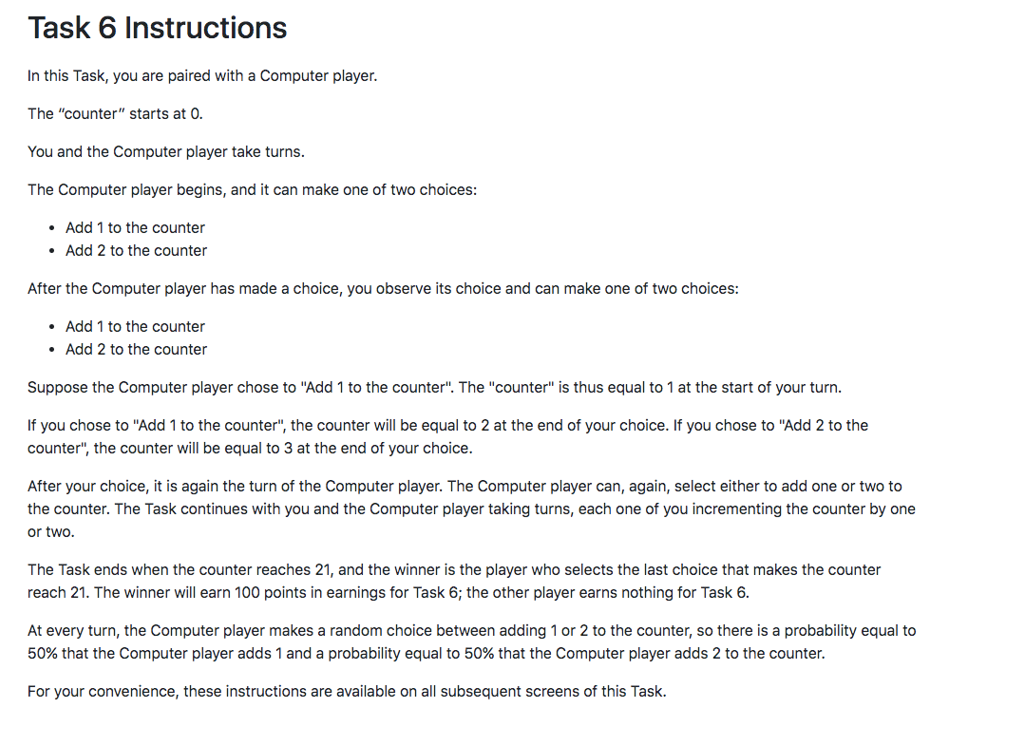}
\end{figure}
\begin{figure}[!hbt]
    \centering
    \includegraphics[width=0.8\linewidth]{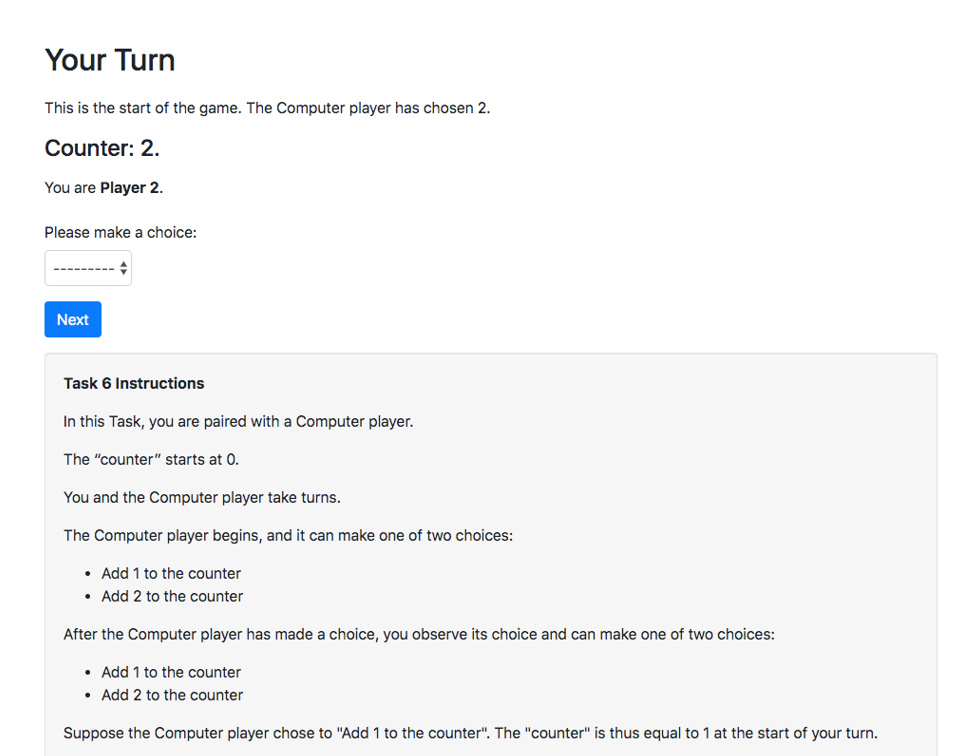}
\end{figure}
\begin{figure}[!hbt]
    \centering
    \includegraphics[width=0.8\linewidth]{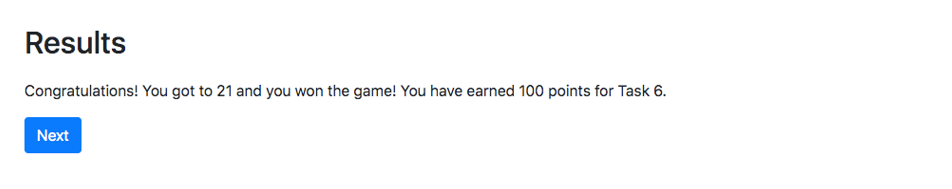}
\end{figure}
\begin{figure}[!hbt]
    \centering
    \includegraphics[width=0.8\linewidth]{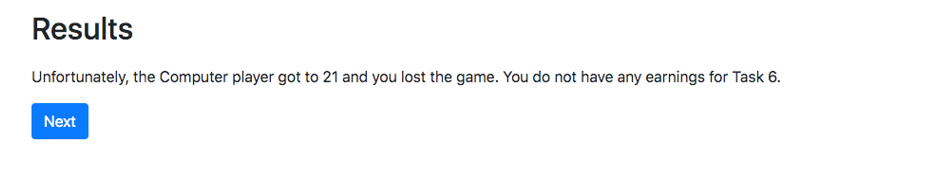}
\end{figure}
\begin{figure}[!hbt]
    \centering
    \includegraphics[width=0.8\linewidth]{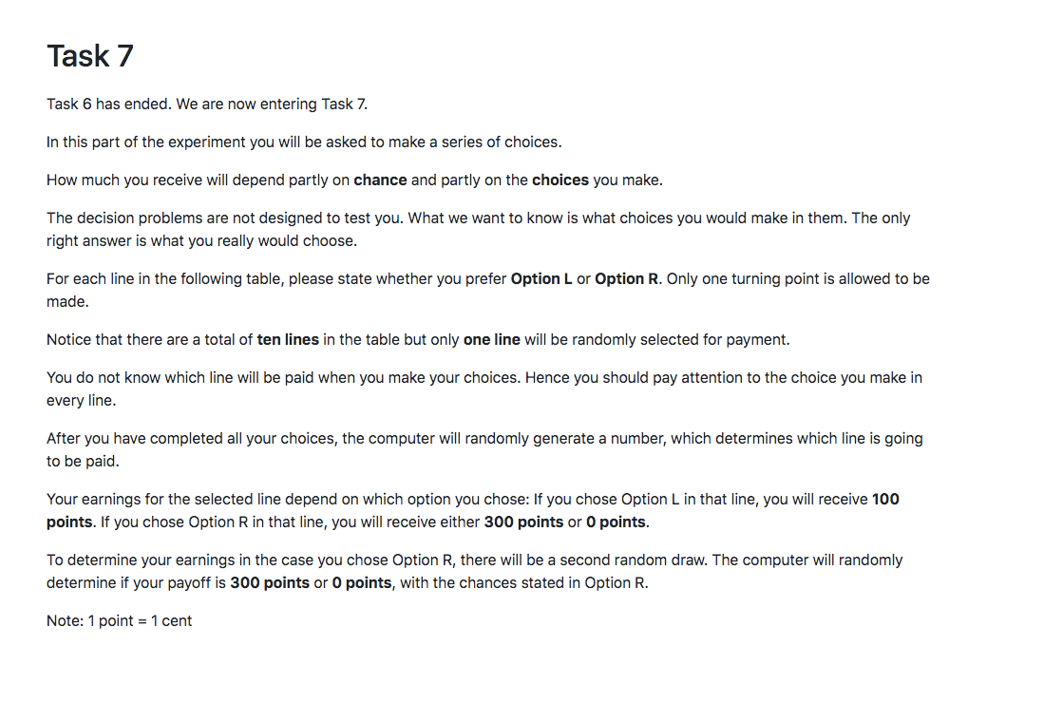}
\end{figure}
\begin{figure}[!hbt]
    \centering
    \includegraphics[width=0.8\linewidth]{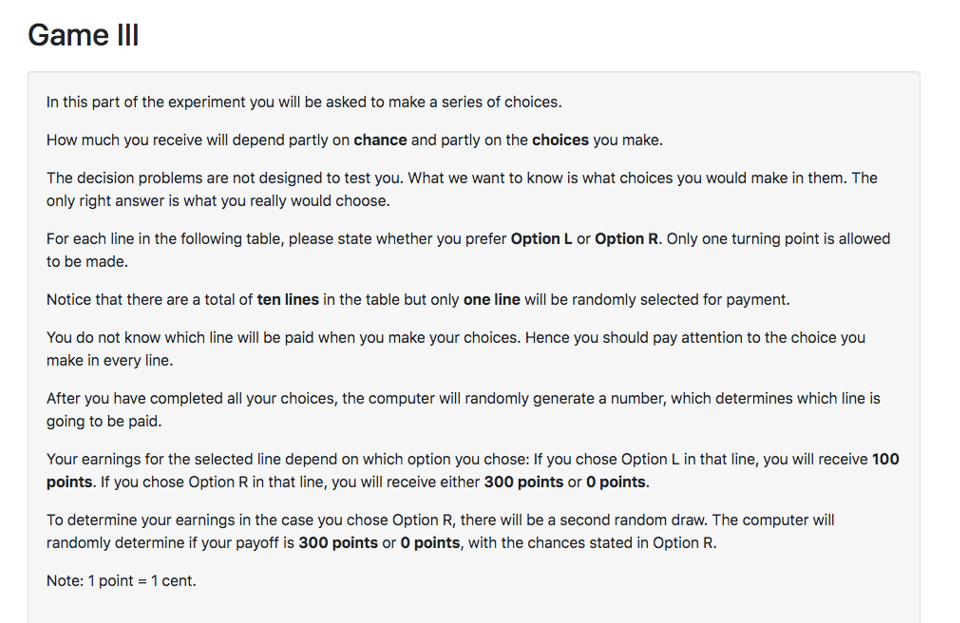}
\end{figure}
\begin{figure}[!hbt]
    \centering
    \includegraphics[width=0.8\linewidth]{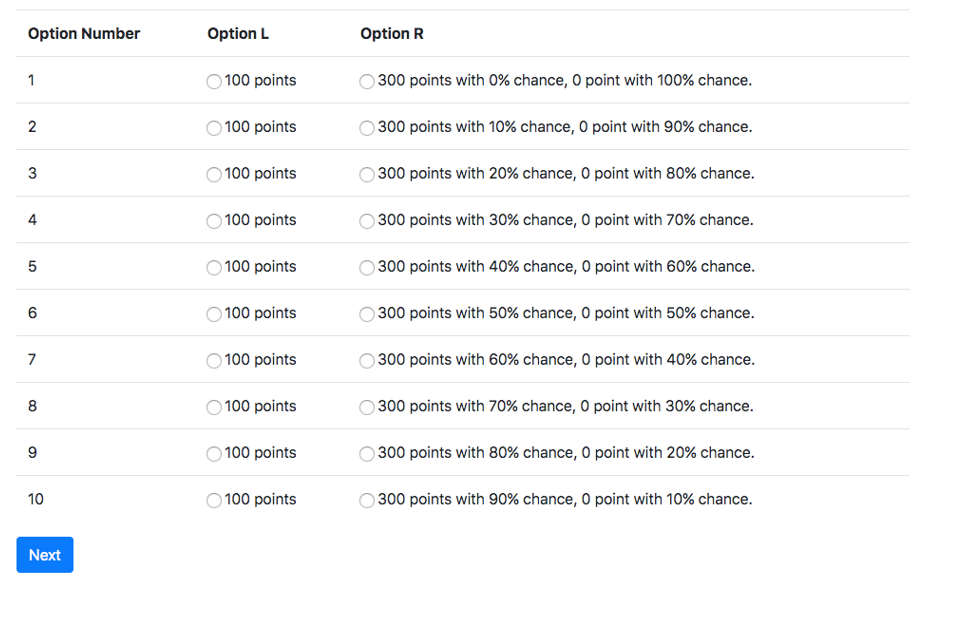}
\end{figure}
\begin{figure}[!hbt]
    \centering
    \includegraphics[width=0.8\linewidth]{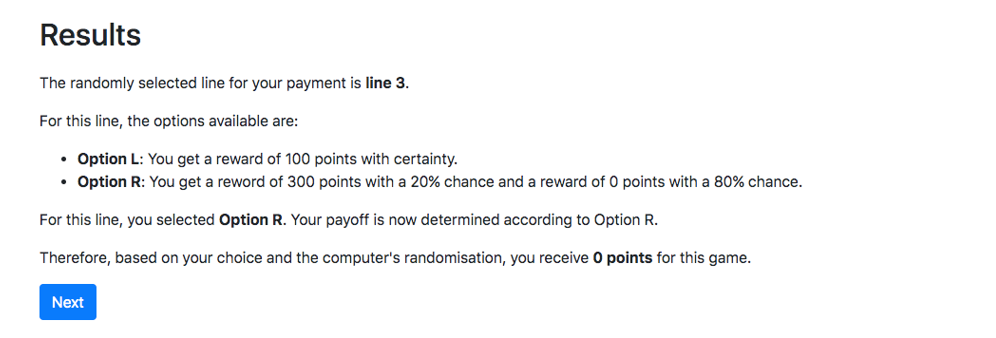}
\end{figure}
\begin{figure}[!hbt]
    \centering
    \includegraphics[width=0.8\linewidth]{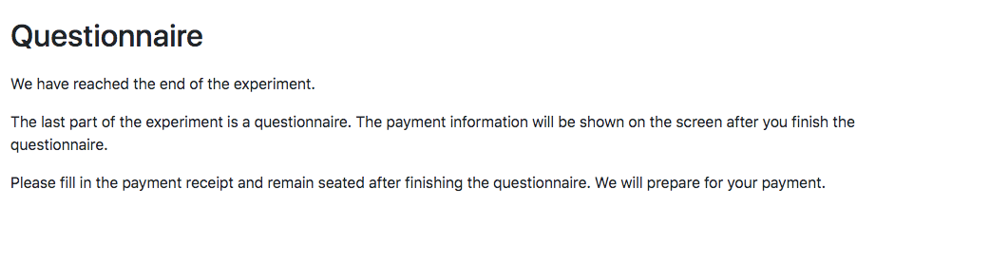}
\end{figure}
\begin{figure}[!hbt]
    \centering
    \includegraphics[width=0.8\linewidth]{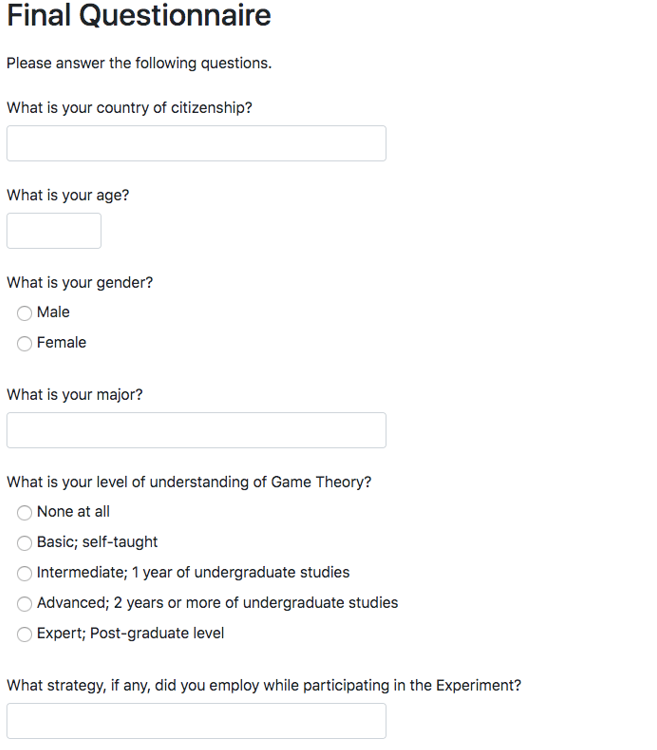}
\end{figure}
\begin{figure}[!hbt]
    \centering
    \includegraphics[width=0.8\linewidth]{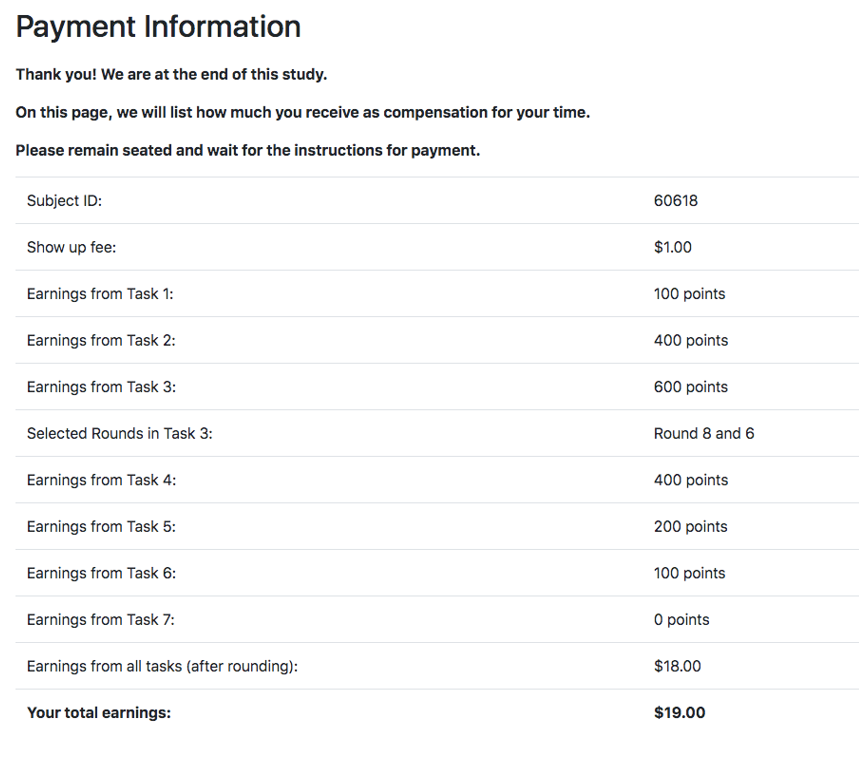}
\end{figure}





\end{document}